\documentclass[11pt,twoside,english]{article}
\usepackage{amsmath}
\usepackage{amsfonts}
\usepackage{mathrsfs}
\usepackage{color}
\usepackage{ amsmath, amsfonts, amssymb, amsthm, amscd}
\usepackage[T1]{fontenc}
\usepackage[english]{babel}
\usepackage[noinfoline]{imsart}

\newtheorem{theorem}{Theorem}

\newtheorem{remark}{Remark}

\newcommand{\be}{\begin{equation}}
\newcommand{\ee}{\end{equation}}
\newcommand{\beq}{\begin{eqnarray}}
\newcommand{\eeq}{\end{eqnarray}}

\newcommand{\ced}{\end{proof}}

\begin{document}

\begin{frontmatter}


\title{Robust control problems of BSDEs coupled with value functions}
\date{}
\runtitle{}
\author{\fnms{Zhou}
 \snm{YANG}\corref{}\ead[label=e1]{yangzhou@scnu.edu.cn}}
\thankstext{T1}{The work of the first author is supported by NSFC (Grant
No.12171169, 11371155, 11801091), and Guangdong Basic and Applied Basic Research Foundation (Grant No.2019A1515011338).}
\address{School of Mathematical Sciences, South China Normal University
\\\printead{e1}}
\author{\fnms{Jing}
 \snm{ZHANG}\corref{}\ead[label=e2]{zhang\_jing@fudan.edu.cn}}
\thankstext{T2}{The work of the second author is partially supported by National Key R\&D Program of China (2018YFA0703900), NSFC (Grant No.12031009) and Shanghai Science and Technology Commission Grant (21ZR1408600).}
\address{School of Mathematical Sciences, Fudan University
\\\printead{e2}}
\author{\fnms{Chao}
 \snm{ZHOU}\corref{}\ead[label=e3]{matzc@nus.edu.sg}}
\thankstext{T1}{The work of the third author is partially supported by the Ministry of Education in Singapore under the grant MOE AcRF R-146-000-271-112 and by NSFC under the grant award 11871364.}
\address{Department of Mathematics and Risk Management Institute, National University of Singapore
\\\printead{e3}}

\runauthor{Z. Yang, J. Zhang and C. Zhou}

%
%

\begin{abstract}
A robust control problem is considered in this paper, where the controlled stochastic differential equations (SDEs) include ambiguity parameters and their coefficients satisfy non-Lipschitz continuous and non-linear growth conditions, the objective function is expressed as a backward stochastic differential equation (BSDE) with the generator depending on the value function. We establish the existence and uniqueness of the value function in a proper space and provide a verification theorem.
Moreover, we apply the results to solve two typical optimal investment problems in the market with ambiguity, one of which is with Heston stochastic volatility model. In particular, we establish some sharp estimations for Heston model with ambiguity parameters.
\end{abstract}




\end{frontmatter}

\section{Introduction}
Since the seminal work of Merton \cite{Merton71}, there have been numerous studies of the optimal portfolio problem for an ambiguity averse investor under various assumptions. Most of these works conclude that it is essential to take ambiguity into account when making a decision for portfolio selection. 
Anderson et al. \cite{AndersonHansenSargent} studied the portfolio choice problems involving uncertainty about the return process for equities. In their setting, model uncertainty is reduced to uncertainty about the drift of the state variables. They used a statistical theory of detection to quantify how much model misspecification the decision maker should fear, given his historical data record.
Maenhout \cite{Maenhout} adapted the general robust control framework of \cite{AndersonHansenSargent} to a dynamic portfolio and consumption problem with power utility. He proposed an important modification, namely homothetic robustness, that preserves wealth independence and analytical tractability. He showed that robustness dramatically decreases the demand for equities and is observationally equivalent to recursive preferences when removing wealth effects. As an extension, he also presented a closed-form solution for the portfolio problem for a robust Duffie-Epstein-Zin investor.
Escobar et al. \cite{Escobar} derived, by using the robust control approach, the optimal portfolio for an ambiguity averse investor in complete market. The stock price follows a stochastic volatility jump-diffusion process and the investor can have different levels of uncertainty about the models for the stock and its volatility. They found that the optimal exposures to stock and volatility risks are significantly affected by the ambiguity aversion to the corresponding risk factor only. However, in incomplete market, the  volatility ambiguity has a smaller impact. The similar problems have also been considered in various 
fields, such as pension funds (see \cite{Baltas2021, Wang2018, Zeng2018}), insurances (see \cite{Yi2013, Zheng2016}) and other related areas (see \cite{NiuYangZhao}). Compared to the previous works, in which the authors intended to find the closed-form solutions and study their properties, we now focus on the mathematical theory as well as the rigorous mathematical proofs.

We also want to mention some works in mathematical finance. Matoussi et al. \cite{Matoussi2015} studied the problem of robust utility maximization in an incomplete market with volatility uncertainty in the framework of second-order backward stochastic differential equations (2BSDEs). They showed that the value function can be written as the initial value of some 2BSDE and proved the existence of an optimal strategy.
Pham et al. \cite{Pham2021} considered a dynamic multi-asset mean-variance portfolio selection problem under model uncertainty and proved a general separation principle for the associated robust control problem.
Jin et al. \cite{JinLuoZeng} solved the optimal portfolio in closed-form to a portfolio choice problem in a multi-asset incomplete market with ambiguous jumps. They found that sizable portfolio rebalancing and subsequent wealth losses are encountered in the presence of jump ambiguity.
In all these papers, the authors considered the classical utility functions rather than the recursive type which is given by a BSDE coupled with value function. Moreover, the state SDEs are still assumed to satisfy the Lipschitz conditions, hence the results cannot be applied to some typically non-Lipschitz case in finance.

When investigating the optimal stochastic control problem involving mean-field BSDEs, in order to give the stochastic interpretation to nonlocal Hamilton-Jacobi-Bellman (HJB) equations of mean-field type, Hao and Li \cite{HaoLi1} introduced a new type of BSDEs coupled with value function. They studied the existence and uniqueness as well as the comparison theorem for such kind of BSDEs coupled with value function and proved the associated nonlocal HJB equation has a unique viscosity solution in the space of continuous functions of at most polynomial growth. Then, they continued to study the fully coupled forward backward stochastic differential equations (FBSDEs) case and the reflected mean-field BSDEs case in \cite{HaoLi2} and \cite{LiLi}, respectively. 
Notice that, in all these papers, the Lipschitz conditions of the coefficients are supposed to hold true. However, in our work, in order to apply the results to two typical examples in finance (see Section 4), we remove the Lipschitz conditions of the coefficients, this makes the problem more difficult, in particular, due to 
the lack of estimates for Heston model.

In this paper, we focus 
on the following robust control problem: the controlled state process is governed by, for $i=1,\cdots,n$ and $0\leq t\leq s\leq T$,
\begin{equation*}
   X^{t,x;\mathbf{u},\mathbf{v}}_{i,\,s}
   =x_i+\int^s_t \alpha_{i}(r,X^{t,x;\mathbf{u},\mathbf{v}}_r,\mathbf{u}_r,\mathbf{v}_r)\,dr
   +\sum\limits_{j=1}^m\int^s_t\beta_{i\,j}
   (r,X^{t,x;\mathbf{u},\mathbf{v}}_r,\mathbf{u}_r,\mathbf{v}_r)\,dB^j_r,
\end{equation*}
the payoff function is of recursive type given by
\begin{eqnarray*}
   &&\hspace{-0.8cm}Y^{t,x;\mathbf{u},\mathbf{v}}_s
   =\int_s^{T\wedge\tau} f\left(r,X_r^{t,x;\mathbf{u},\mathbf{v}},Y^{t,x;\mathbf{u},\mathbf{v}}_r,
   Z^{t,x;\mathbf{u},\mathbf{v}}_r,
   W\left(r,X_r^{t,x;\mathbf{u},\mathbf{v}}\right),\mathbf{u}_r,\mathbf{v}_r\right)dr
   \nonumber\\
   &&\quad\quad-\sum_{j=1}^m\int_s^{T\wedge\tau} Z^{t,x;\mathbf{u},\mathbf{v}}_{j,r}\,dB^j_r
   +\Psi\left(T\wedge\tau,X_{T\wedge\tau}^{t,x; \mathbf{u}, \mathbf{v}}\right),
   \quad 0\leq t\leq s\leq T\wedge\tau,
\end{eqnarray*}
where $\tau$ is the first escaping time and the value function is defined as follows
\begin{equation*}
    W(t,x)=\sup\limits_{\mathbf{u}\in{\cal U}(t,x)}\inf\limits_{\mathbf{v}\in{\cal V}(t,x)}
    Y^{t,x;\mathbf{u},\mathbf{v}}_t.
\end{equation*}
We prove that the above stochastic control problem admits a unique value function in some proper space under some mild conditions by the contraction mapping method and the bootstrap method.
As the generator of the BSDEs contains the value function, we first fix the value function in the generator, then our control problem becomes as a standard stochastic control problem. By using the localization method, together with dynamic programming principle (DPP) for the standard control problem and the relation between the standard control problem and our original problem, we prove the verification theorem, which states that the classical solution to HJBI equation is the unique value function to our robust control problem.
Our results can be applied to solve some typical optimal investment problems in the market with ambiguity, one of which is with Heston stochastic volatility model. In particular, some sharp estimations for Heston model with ambiguity parameters are established.
\vspace{1mm}

The main difficulties of our work come from the following four aspects: first, the objective functional is given by a BSDE coupling with the value function; second, the lack of estimates for Heston model in a market with ambiguity, in fact, to the best of our knowledge, even in a  market with non-ambiguity, there are few results on the estimates of Heston model;  third, the controlled SDEs are no longer assumed to satisfy the Lipschitz condition, so that the results can be applied to various cases, for example, to the Heston model; 
last, a proper space of the value functions need to be given.
\vspace{1mm}

One of the key points in our paper is the local estimate for the Heston model with ambiguity parameters (Conclusion (1) in Theorem A1), which may be   powerful for some stochastic control problems of Heston model. For example, in \cite{Kraft}, Kraft et al. assumed that the speed of mean reversion in volatility divided by the bound of the admissible investment strategy is large enough in their optimal investment consumption problem under Heston model. Futhermore, Pu and Zhang \cite{PuZhang} required that the admissible ambiguity parameters must take some special forms in the robust optimal investment consumption problem under Heston model with ambiguity. However, these conditions are completely weakened in our paper by means of the local estimate, the bootstrap method and the DPP, where the model is more generalized than those in \cite{Kraft} and \cite{PuZhang}.
\vspace{1mm}

The rest of this paper is organized as follows. In the next section, we present our stochastic control problem and set the notations and assumptions. Section 3 is devoted to establishing the existence and uniqueness of the value function as well as a verification theorem which shows the relationship between the stochastic control problem (SCP) and its associate HJBI equation. Then in Section 4, we discuss two applications in finance of our main results. Finally, we give the estimates for Heston model and the proof of Theorem \ref{result2 of example 2} in Appendix A and Appendix B, respectively.

\section{Preliminaries}

\subsection{The state equation}

The state $X=(X_1,\cdot\cdot\cdot,X_n)^{\rm T},\,n\in \mathbb{Z}^+$ is given by the following controlled stochastic differential equations, for $i=1,\cdot\cdot\cdot,n$ and $0\leq t\leq s\leq T$,
\begin{equation}\label{stateequation1}
   X^{t,x;\mathbf{u},\mathbf{v}}_{i,\,s}
   =x_i+\int^s_t \alpha_{i}(r,X^{t,x;\mathbf{u},\mathbf{v}}_r,\mathbf{u}_r,\mathbf{v}_r)\,dr
   +\sum\limits_{j=1}^m\int^s_t\beta_{i\,j}
   (r,X^{t,x;\mathbf{u},\mathbf{v}}_r,\mathbf{u}_r,\mathbf{v}_r)\,dB^j_r.
\end{equation}
Here  $m\in\mathbb{Z}^+$ and $T\in(0,+\infty)$ is fixed. The initial value $x\in {\cal X}$, where ${\cal X}$ is an open domain in $\mathbb{R}^n$ with continuous boundary $\partial {\cal X}$. $\mathbf{u}$ and $\mathbf{v}$ are control processes. $B$ is an $m$-dimensional standard Brownian motion on the probability space $(\Omega,{\cal F},\mathbb{F},\mathbb{P})$ satisfying the usual conditions.
Without loss of generality, suppose that $\mathbb{F}$ is the completed natural filtration generated by the Brownian motion $B$.

Let ${\cal N}$ denote the class of $\mathbb{P}-$null sets in ${\cal F}$ and $\sigma_1\vee\sigma_2$ denote the $\sigma-$field generated by $\sigma_1\cup\sigma_2$. For given $t\in[0,T)$, the augmented filtration $\{{\cal F}^t_s\}_{t\leq s\leq T}$ is well-defined with ${\cal F}^t_s\triangleq \sigma(\{B_r-B_t\}_{t\leq r\leq s})\vee{\cal N}$. Particularly, we write ${\cal F}_s={\cal F}^0_s$.

For given $t\in[0,T)$, denote by ${\cal U}[t,T]$ and ${\cal V}[t,T]$ the sets of progressively measurable processes with respect to $\{{\cal F}^t_s\}_{t\leq s\leq T}$ which take values respectively in sets ${\cal U}\subset\mathbb{R}^N$ and ${\cal V}\subset\mathbb{R}^M$ with fixed positive integers $N,M$. The coefficients $\alpha_i,\beta_{ij}:[0,T]\times \mathbb{R}^n\times{\cal U}\times{\cal V}\rightarrow\mathbb{R}\,(i=1,\cdots,n,\,j=1,\cdots,m)$ are Borel-measurable functions.

\subsection{Notations}

We use the following notations throughout this paper:
\vspace{2mm}

$\bullet\;\;{\cal X}_t^s\triangleq[t,s]\times {\cal X},\;\;{\cal X}^s\triangleq[0,s)\times {\cal X}$, for $0\leq t\leq s\leq T$;

$\bullet\;\;\partial{\cal X}^T\triangleq\{T\}\times\overline{{\cal X}}\cup[0,T]\times\partial {\cal X}$, the parabolic boundary of ${\cal X}^T$;

$\bullet\;\; C(\overline{{\cal X}^T})$, the set of all continuous functions on $\overline{{\cal X}^T}$;

$\bullet\;\; C^{1,2}({\cal X}^T)$, the set of all functions whose first derivative with respect to $t$ and second derivative with respect to $x$ are continuous on ${\cal X}^T$;

$\bullet\;\;{\cal B}_{(r,s)}$, the set of all Borel-measurable functions with norm
$$
 \|w\|_{{\cal B}_{(r,s)}}\triangleq
 {\rm ess}\!\sup\left\{{|w(t,x)|\over b_1(t,x)}:(t,x)\in {\cal X}_r^s\right\},
$$
where $b_1$ is some poisitive Borel-measurable function.
\vspace{2mm}

For $q\geq1$, we define
\vspace{2mm}

$\bullet\;\;L^q_{s}$, the set of all ${\cal F}_s-$measurable random variables with norm
$
 \|\xi\|_{L^q_{s}}\triangleq [\mathbb{E}(|\xi|^q)]^{1/q};
$

$\bullet\;\;\mathbb{L}^q(t,s)$, the set of all $\{{\cal F}^t_s\}_{t\leq s\leq T}$-progressively measurable processes with norm
$$
 \|X\|_{\mathbb{L}^q(t,s)}\triangleq \bigg[\,\mathbb{E}\Big(\int_t^s\,|\,X_r\,|^q\,dr\Big)\, \bigg]^{1\over q};
$$

$\bullet\;\;\mathbb{S}^q(t,s)$, the set of all continuous processes in $\mathbb{L}^q(t,s)$ with norm
$$
 \|X\|_{\mathbb{S}^q(t,s)}\triangleq
 \bigg[\,\mathbb{E}
 \Big(\sup\limits_{r\in[\,t,\,s\,]}|\,X_r\,|^q\Big)\,\bigg]^{1\over q}.
$$

\subsection{The payoff function}

We consider the payoff function introduced by the following BSDE:
\begin{eqnarray}\label{payoff1}
   &&\hspace{-0.8cm}Y^{t,x;\mathbf{u},\mathbf{v}}_s
   =\int_s^{T\wedge\tau} f\left(r,X_r^{t,x;\mathbf{u},\mathbf{v}},Y^{t,x;\mathbf{u},\mathbf{v}}_r,
   Z^{t,x;\mathbf{u},\mathbf{v}}_r,
   W\left(r,X_r^{t,x;\mathbf{u},\mathbf{v}}\right),\mathbf{u}_r,\mathbf{v}_r\right)dr
   \nonumber\\
   &&\quad\quad-\sum_{j=1}^m\int_s^{T\wedge\tau} Z^{t,x;\mathbf{u},\mathbf{v}}_{j,r}\,dB^j_r
   +\Psi\left(T\wedge\tau,X_{T\wedge\tau}^{t,x; \mathbf{u}, \mathbf{v}}\right),
   \quad 0\leq t\leq s\leq T\wedge\tau,
\end{eqnarray}
where $\tau\triangleq\inf\{s\geq t:X_s^{t,x;\mathbf{u},\mathbf{v}}\notin{\cal X}\}$ and
\begin{equation}\label{valuefunction1}
    W(t,x)=\sup\limits_{\mathbf{u}\in{\cal U}(t,x)}\inf\limits_{\mathbf{v}\in{\cal V}(t,x)}
    Y^{t,x;\mathbf{u},\mathbf{v}}_t.
\end{equation}
Here ${\cal U}(t,x)$ and ${\cal V}(t,x)$ are the admissible sets which will be defined later. In BSDE~\eqref{payoff1}, the generator $f:\overline{{\cal X}^T}\times\mathbb{R}\times\mathbb{R}^m\times\mathbb{R}\times {\cal U}\times{\cal V}\rightarrow\mathbb{R}$ and the terminal condition $\Psi:\overline{{\cal X}^T}\rightarrow\mathbb{R}$ are Borel-measurable functions that satisfy the following assumptions:
\vspace{2mm}

\textbf{Assumption B1.} $f$ is Lipschitz continuous with respect to $(y,z,w)$, i.e., there exist a constant $K_1>0$ and a non-negative Borel-measurable function $b_2:[0,T]\times\mathcal{U}\times\mathcal{V}\rightarrow\mathbb{R}$ such that
$$
   |f(t,x,y,z,w,u,v)-f(t,x,y^\prime,z^\prime,w^\prime,u,v)|\leq
   K_1(|y-y^\prime|+|z-z^\prime|)+b_2(t,u,v)|w-w^\prime|,
$$
for any $(t,x,u,v)\in\overline{{\cal X}^T}\times{\cal U}\times{\cal V},\ y,y^\prime\in \mathbb{R},\ z,z^\prime\in\mathbb{R}^m\ \mbox{and}\ w,w^\prime\in\mathbb{R}$.
\vspace{2mm}

\textbf{Assumption B2. }There exist a positive Borel-measurable function $b_1:\overline{{\cal X}^T}\rightarrow\mathbb{R}$ and a non-negative Borel-measurable function $b_3:[0,T]\times\mathcal{U}\times\mathcal{V}\rightarrow\mathbb{R}$ such that
$$
   |f(t,x,0,0,0,u,v)|\leq b_1(t,x)b_2(t,u,v)+b_3(t,u,v),\;\forall\;(t,x,u,v)\in\overline{{\cal X}^T}
   \times{\cal U}\times{\cal V}.
$$

\textbf{Assumption B3.} $\Psi\in C(\,\overline{{\cal X}^T}\,)$ and there exists a non-negative Borel-measurable function $b_4:\overline{\mathcal{X}}\rightarrow\mathbb{R}$ such that
$$
   |\Psi(t,x)|\leq b_4(x),\;\;\forall\;(t,x)\in\overline{{\cal X}^T}.
$$

\begin{remark}
The difference between SCP \eqref{valuefunction1} and the standard SCP is that the payoff function depends on the value function $W$.
\end{remark}\smallskip

\subsection{The admissible sets}
We denote by ${\cal U}(t,x)$ and ${\cal V}(t,x)$ the admissible sets of controls $\mathbf{u}$ and $\mathbf{v}$, which are nonempty subsets of ${\cal U}[t,T]$ and ${\cal V}[t,T]$, respectively and make the following assumptions:
\vspace{0.2mm}

\textbf{Assumption A1. (``Switching condition'')} For any $(t,x)\in {\cal X}^T,\,\mathbf{u}\in {\cal U}(t,x)$ and $\mathbf{v}\in {\cal V}(t,x)$, SDE \eqref{stateequation1} has a unique strong solution $X^{t,x;\mathbf{u},\mathbf{v}}$.
Moreover, for any $\{{\cal F}^t_s\}_{t\leq s\leq T}-$stopping time ${\underline{\tau}}$ satisfying $t\leq {\underline{\tau}}\leq\inf(s\geq t:X_s^{t,x;\mathbf{u},\mathbf{v}}\notin{\cal X})\wedge T$, let $\widetilde{\mathbf{u}}\in {\cal U}({\underline{\tau}},X_{\underline{\tau}}^{t,x;\mathbf{u},\mathbf{v}}),
\widetilde{\mathbf{v}}\in{\cal V}({\underline{\tau}},X_{\underline{\tau}}^{t,x;\mathbf{u},\mathbf{v}})$, define new controls $\widehat{\mathbf{u}}$ and $\widehat{\mathbf{v}}$ by
$$
   \widehat{\mathbf{u}}_s=
   \left\{
   \begin{array}{ll}
   \mathbf{u}_s,&t\leq s\leq{\underline{\tau}},
   \vspace{2mm}\\
   \widetilde{\mathbf{u}}_s,&{\underline{\tau}}\leq s\leq\widetilde{\tau},
   \end{array}
   \right.
 \hspace{1cm}  \widehat{\mathbf{v}}_s=
   \left\{
   \begin{array}{ll}
   \mathbf{v}_s,&t\leq s\leq{\underline{\tau}},
   \vspace{2mm}\\
   \widetilde{\mathbf{v}}_s,&{\underline{\tau}}\leq s\leq\widetilde{\tau},
   \end{array}
   \right.
$$
where
$$ \widetilde{\tau}=\inf\left(s\geq {\underline{\tau}}:X_s^{{\underline{\tau}},
   X_{\underline{\tau}}^{t,x;\mathbf{u},\mathbf{v}};
   \widetilde{\mathbf{u}},\widetilde{\mathbf{v}}}
   \notin{\cal X}\right)\wedge T.$$
Then, for any $\{{\cal F}^t_s\}_{t\leq s\leq T}-$stopping time $\overline{\tau}$ satisfying $t\leq\overline{\tau}\leq \widetilde{\tau}$, one has $\widehat{\mathbf{u}}|_{(\overline{\tau},\widetilde{\tau})}
\in {\cal U}
({\overline{\tau}},X_{\overline{\tau}}^{t,x;\widehat{\mathbf{u}},\widehat{\mathbf{v}}})$ and
$\widehat{\mathbf{v}}|_{(\overline{\tau},\widetilde{\tau})}
\in {\cal V}
({\overline{\tau}},X_{\overline{\tau}}^{t,x;\widehat{\mathbf{u}},\widehat{\mathbf{v}}})$,  where $\widehat{\mathbf{u}}|_{(\overline{\tau},\widetilde{\tau})}$ and $\widehat{\mathbf{v}}|_{(\overline{\tau},\widetilde{\tau})}$ denote respectively the restrictions of $\widehat{\mathbf{u}}$ and $\widehat{\mathbf{v}}$ to $[\overline{\tau},\widetilde{\tau}]$.
\vspace{2mm}

\textbf{Assumption A2(N). }There exist a positive constant $K_2$ and a positive integer $N$ such that, for any $(t,x)\in{\cal X}_{(i-1)T/N}^{iT/N},\ i=1,2,\cdots,N,\,\mathbf{u}\in {\cal U}(t,x)$ and $\mathbf{v}\in {\cal V}(t,x)$, the strong solution $X^{t,x;\mathbf{u},\mathbf{v}}$ to SDE \eqref{stateequation1} satisfies
$$
   \left\|b_1\left(\cdot,X^{t,x;\mathbf{u},\mathbf{v}}\right)b_2(\cdot,\mathbf{u},\mathbf{v})
   +b_3(\cdot,\mathbf{u},\mathbf{v})\right\|_{\mathbb{L}^q(t,iT/N)}
   +\left\|b_4\left(X^{t,x;\mathbf{u},\mathbf{v}}\right)\right\|_{\mathbb{S}^q(t,iT/N)}
   \leq K_2b_1(t,x).
$$


\textbf{Assumption A3(N). }For any $i=1,2,\cdots,N$, there exists a $\epsilon>0$ such that
$$
 \sup\limits_{\mathbf{u}\in {\cal U}(t,x),\mathbf{v}\in {\cal V}(t,x),(t,x)\in{\cal X}_{(i-1)T/N}^{iT/N}}\!\!\!\!
 {\left\|b_1\left(\cdot,X^{t,x;\mathbf{u},\mathbf{v}}\right)b_2(\cdot,\mathbf{u},\mathbf{v})
 \right\|_{\mathbb{L}^q(t,(t+\epsilon)\wedge(iT/N))}\over b_1(t,x)}
 \leq{e^{-(K_1+1)^2T/(2N)}\over 60}.
$$
\begin{remark}
Assumption A1 is important for the standard SCP to satisfy the DPP. In Assumption A2(1) and A3(1) (i.e., $N=1$), the inequalities hold on the whole domain ${\cal X}^T$, whereas, in Assumption A2(N) and A3(N) with $N\geq 2$, the inequalities hold only on the local domain ${\cal X}_{(i-1)T/N}^{iT/N}$, $i=1,2,\cdots,N$.
\end{remark}

\section{The main results}
\subsection{The existence and uniqueness of value function}
In this subsection, we give a result on the existence and uniqueness of the value function for
SCP \eqref{valuefunction1} in some proper space.

\begin{theorem}\label{maintheorem1}
Suppose that Assumption B1, B2, B3 and A1, A2(1) and A3(1) with $q=2$ hold. Then SCP \eqref{valuefunction1} has a unique value function $W$ in the space ${\cal B}_{(0,T)}$.
\end{theorem}
\begin{remark}\label{remark of main result1}
If the generator $f$ is independent of $y$ and $z$, i.e.,
$$
   Y^{t,x;\mathbf{u},\mathbf{v}}_t
   =\mathbb{E}\left[\Psi\left(T\wedge\tau,X_{T\wedge\tau}^{t,x; \mathbf{u}, \mathbf{v}}\right)
   +\int_t^{T\wedge\tau} f\left(r,X_r^{t,x;\mathbf{u},\mathbf{v}},
   W\left(r,X_r^{t,x;\mathbf{u},\mathbf{v}}\right),\mathbf{u}_r,\mathbf{v}_r\right)\,dr\right].
$$
Then the conclusion in Theorem \ref{maintheorem1} still holds if Assumption B1, B2, B3 and A1, A2(1) and A3(1) with $q=1$ are satisfied. This can be deduced similarly as the following proof.
\end{remark}

\begin{proof}
Choose $\epsilon>0$ small enough which will be defined later. We first consider SCP \eqref{valuefunction1} on $[T-\epsilon,T]$. Fix a function $J\in{\cal B}_{(T-\epsilon,T)}$, denote
\begin{equation}\label{newfunction1}
   f^J(r,x,y,z,u,v)\triangleq f(r,x,y,z,J(r,x),u,v)
\end{equation}
and construct an auxiliary standard SCP:
\begin{equation}\label{valuefunction2}
    W^J(t,x)=\sup\limits_{\mathbf{u}\in{\cal U}(t,x)}\inf\limits_{\mathbf{v}\in{\cal V}(t,x)}
    Y^{J;t,x;\mathbf{u},\mathbf{v}}_t,\quad\forall\;(t,x)\in {\cal X}_{T-\epsilon}^{T}\,,
\end{equation}
where $Y^{J;t,x;\mathbf{u},\mathbf{v}}$ is determined by the following BSDE:
\begin{eqnarray}\label{payoff2}
   Y^{J;t,x;\mathbf{u},\mathbf{v}}_s
   &=&\Psi\left(T\wedge\tau,X_{T\wedge\tau}^{t,x;\mathbf{u},\mathbf{v}}\right)
   +\int_s^{T\wedge\tau} f^J\left(r,X_r^{t,x;\mathbf{u},\mathbf{v}},Y^{J;t,x;\mathbf{u},\mathbf{v}}_r,
   Z^{J;t,x;\mathbf{u},\mathbf{v}}_r,\mathbf{u}_r,\mathbf{v}_r\right)\,dr
   \nonumber\\
   &&-\sum_{j=1}^m\int_s^{T\wedge\tau} Z^{J;t,x;\mathbf{u},\mathbf{v}}_{j,\,r}\,dB^j_r,\qquad
   0\leq t\leq s\leq T\wedge\tau,
\end{eqnarray}
and $X^{t,x;\mathbf{u},\mathbf{v}}$ is subject to~\eqref{stateequation1}.

From Assumption B1, we know that the generator $f^J$ of BSDE~\eqref{payoff2} is Lipschitz continuous with respect to $(y,z)$. Combining with Assumption B1, B2, B3 and A2(1), we deduce
\begin{eqnarray*}
   &&\left\|f^J(\cdot,X^{t,x;\mathbf{u},\mathbf{v}},0,0,\mathbf{u},
   \mathbf{v})\right\|_{\mathbb{L}^2(t,T)}
   \nonumber\\[2mm]
   &\leq& \left(1+\|J\|_{{\cal B}_{(T-\epsilon,T)}}\right)
   \left\|b_1\left(\cdot,X^{t,x;\mathbf{u},\mathbf{v}}\right)
   b_2(\cdot,\mathbf{u},\mathbf{v})
   +b_3(\cdot,\mathbf{u},\mathbf{v})\right\|_{\mathbb{L}^2(t,T)}
   \nonumber\\[2mm]
   &\leq& K_2\left(1+\|J\|_{{\cal B}_{(T-\epsilon,T)}}\right)b_1(t,x)
\end{eqnarray*}
and
\begin{equation*}
   \left\|\Psi\left(T\wedge\tau,
   X_{T\wedge\tau}^{t,x;\mathbf{u},\mathbf{v}}\right)\right\|_{L^2_{T\wedge\tau}}
   \leq \left\|b_4\left(X_{T\wedge\tau}^{t,x;\mathbf{u},\mathbf{v}}\right)\right\|_{L^2_{T\wedge\tau}}
   \leq \left\|b_4\left(X^{t,x;\mathbf{u},\mathbf{v}}\right)\right\|_{\mathbb{S}^2(t,T)}
   \leq K_2b_1(t,x).
\end{equation*}
So, BSDE~\eqref{payoff2} admits a unique solution which satisfies the following estimate:
\begin{equation*}
   \left\|Y^{J;t,x;\mathbf{u},\mathbf{v}}\right\|_{\mathbb{S}^2(t,T)}
   +\left\|Z^{J;t,x;\mathbf{u},\mathbf{v}}\right\|_{\mathbb{L}^2(t,T)}
   \leq Cb_1(t,x),
\end{equation*}
where the constant $C$ depends only on $T,K_1,K_2$ and $\|J\|_{{\cal B}_{(T-\epsilon,T)}}$. 
Thus, $W^J$ is well-defined and $\forall (t,x)\in\mathcal{X}^T_{T-\epsilon}$, $|W^J(t,x)|\leq Cb_1(t,x)$, which implies $W^J\in{\cal B}_{(T-\epsilon,T)}$.

Then, we define ${\cal T}: J\mapsto W^J$ and prove that ${\cal T}$ is a contraction mapping. For any $J_1,J_2\in{\cal B}_{(T-\epsilon,T)}$ and $\mathbf{u}\in{\cal U}(t,x),\mathbf{v}\in{\cal V}(t,x)$, denote
$$
   \Delta Y\triangleq Y^{J_1;t,x;\mathbf{u},\mathbf{v}}-Y^{J_2;t,x;\mathbf{u},\mathbf{v}},\qquad \Delta Z \triangleq Z^{J_1;t,x;\mathbf{u},\mathbf{v}}-Z^{J_2;t,x;\mathbf{u},\mathbf{v}},
$$
and
\begin{eqnarray*}
    \Delta f_r&\triangleq &f^{J_1}\left(r,X_r^{t,x;\mathbf{u},\mathbf{v}},
    Y^{J_1;t,x;\mathbf{u},\mathbf{v}}_r,
    Z^{J_1;t,x;\mathbf{u},\mathbf{v}}_r,\mathbf{u}_r,\mathbf{v}_r\right)\\&&
    -f^{J_2}\left(r,X_r^{t,x;\mathbf{u},\mathbf{v}},
    Y^{J_2;t,x;\mathbf{u},\mathbf{v}}_r,
    Z^{J_2;t,x;\mathbf{u},\mathbf{v}}_r,\mathbf{u}_r,\mathbf{v}_r\right).
\end{eqnarray*}
Then $(\Delta Y,\Delta Z)$ satisfies the following BSDE:
\begin{equation*}
   \Delta Y_s=\int_{s}^{T\wedge\tau} \Delta f_r\,dr
   -\sum\limits_{j=1}^m\int_{s}^{T\wedge\tau}\Delta Z_{j,\,r}\,dW^j_r,\quad
   \forall\;t\leq s\leq T\wedge\tau.
\end{equation*}

By Assumption B1, we have
\begin{eqnarray*}
   |\Delta f_r|&\leq&
   K_1\left(|\Delta Y_r|+|\Delta Z_r|\right)
   +b_2(r,\mathbf{u}_r,\mathbf{v}_r)
   \left|(J_1-J_2)\left(r,X_r^{t,x;\mathbf{u},\mathbf{v}}\right)\right|
   \\[2mm]
   &\leq& K_1\left(|\Delta Y_r|+|\Delta Z_r|\right)
   +b_2(r,\mathbf{u}_r,\mathbf{v}_r)
   \|J_1-J_2\|_{{\cal B}_{(T-\epsilon,T)}}
   b_1\left(r,X_r^{t,x;\mathbf{u},\mathbf{v}}\right).
\end{eqnarray*}

Applying the standard estimations for BSDEs, it is not difficult to derive that, for any $\mathbf{u}\in {\cal U}(t,x),\ \mathbf{v}\in {\cal V}(t,x)$ and $(t,x)\in {\cal X}_{T-\epsilon}^{T}$,
$$
   \|\Delta Y\|_{\mathbb{S}^2(t,T)}
   \leq 30e^{(K_1+1)^2T/2}\|J_1-J_2\|_{{\cal B}_{(T-\epsilon,T)}}
   \left\|b_1\left(\cdot,X^{t,x;\mathbf{u},\mathbf{v}}\right)
   b_2(\cdot,\mathbf{u},\mathbf{v})\right\|_{\mathbb{L}^2(t,T)}.
$$
Then from Assumption A3(1), we know that
$$
   30e^{(K_1+1)^2T/2}\left\|b_1\left(\cdot,X^{t,x;\mathbf{u},\mathbf{v}}\right)
   b_2(\cdot,\mathbf{u},\mathbf{v})\right\|_{\mathbb{L}^2(t,T)}
   \leq {1\over2}b_1(t,x),
$$ for any $\mathbf{u}\in {\cal U}(t,x),\ \mathbf{v}\in {\cal V}(t,x)$ and $(t,x)\in {\cal X}_{T-\epsilon}^{T}$.
Therefore, we obtain
\begin{eqnarray*}
   &&\|{\cal T}(J_1)-{\cal T}(J_2)\|_{{\cal B}_{(T-\epsilon,T)}}
   =\left\|W^{J_1}-W^{J_2}\right\|_{{\cal B}_{(T-\epsilon,T)}}
   \nonumber\\[2mm]
   &\leq&\sup\limits_{\mathbf{u}\in {\cal U}(t,x),\,\mathbf{v}\in{\cal V}(t,x),
   \,(t,x)\in {\cal X}_{T-\epsilon}^{T}}{|\Delta Y_t|\over b_1(t,x)}
   \leq {1\over2}\|J_1-J_2\|_{{\cal B}_{(T-\epsilon,T)}}.
\end{eqnarray*}
Until now, we have showed that ${\cal T}$ is a contraction mapping, hence, it has a unique fixed point, which implies that there exists a unique value function $W$ of SCP \eqref{valuefunction1} on $[T-\epsilon,T]$.

\vspace{2mm}
Next, we extend the result to the whole interval $[0,T]$. Suppose that we have proved there exists a unique value function $W$ of the SCP on ${\cal X}_{\widehat{T}}^{T}$. Fix a function $J\in{\cal B}_{(\widehat{T}-\epsilon,\widehat{T})}$, where $\epsilon$ is chosen as the above. Without loss of generality, we assume $\epsilon\leq\widehat{T}$. Denote
$$
   f^J(r,x,y,z,u,v)\triangleq
   f\left(r,x,y,z,J(r,x)I_{\{\,\widehat{T}-\epsilon\leq r<\widehat{T}\,\}}
   +W(r,x)I_{\{\,\widehat{T}\leq r\leq T\}},u,v\right)
$$
and consider the standard SCP \eqref{valuefunction2} on ${\cal X}_{\widehat{T}-\epsilon}^{T}$.

Define ${\cal T}(J)$ be the restriction of the value function $W^J$ on ${\cal X}_{\widehat{T}-\epsilon}^{\widehat{T}}$. Repeating the same argument as the above, we obtain
\begin{eqnarray*}
   \|{\cal T}(J_1)-{\cal T}(J_2)\|_{{\cal B}_{(\widehat{T}-\epsilon,\widehat{T})}}&=&
   \left\|W^{J_1}-W^{J_2}\right\|_{{\cal B}_{(\widehat{T}-\epsilon,\widehat{T})}}
   \\[2mm]
   &\leq& \sup\limits_{\mathbf{u} \in {\cal U}(t,x),\,\mathbf{v}\in{\cal V}(t,x),\,
   (t,x)\in{\cal X}_{\widehat{T}-\epsilon}^{\widehat{T}}}
   {|\Delta Y_t|\over b_1(t,x)}
   \leq{1\over 2}\|J_1-J_2\|_{{\cal B}_{(\widehat{T}-\epsilon,\widehat{T})}},
\end{eqnarray*}

which provides that ${\cal T}$ is a contraction mapping. Hence, it admits a unique fixed point. Denote the fixed point by $\widetilde{W}$ and extend $W$ onto ${\cal X}_{\widehat{T}-\epsilon}^{\widehat{T}}$ by means of $W=\widetilde{W}$. Then $W$ is exactly the value function $W^W$ on ${\cal X}_{\widehat{T}-\epsilon}^{T}$ of the standard SCP \eqref{valuefunction2} with $J=W$,
which implies that $W$ is also a value function on ${\cal X}_{\widehat{T}-\epsilon}^{T}$ of the original SCP \eqref{valuefunction1}.
\vspace{2mm}

Now, we come to prove the uniqueness of the value function in space ${\cal B}_{(\widehat{T}-\epsilon,T)}$ of problem \eqref{valuefunction1} on the domain ${\cal X}_{\widehat{T}-\epsilon}^{T}$. Assume $W_1,W_2\in{\cal B}_{(\widehat{T}-\epsilon,T)}$ are the value functions of SCP \eqref{valuefunction1} on ${\cal X}_{\widehat{T}-\epsilon}^{T}$. Then $W_1$ is the value function on ${\cal X}_{\widehat{T}-\epsilon}^{T}$ of the standard SCP \eqref{valuefunction2} with $J=W_1$. Thank to Assumption A1 and the DPP for the standard SCP \eqref{valuefunction2}, we know that $W_1$ is the unique value function on ${\cal X}_{\widehat{T}}^{T}$ of SCP \eqref{valuefunction2} with $J=W_1$, which is just the unique value function $W$ on ${\cal X}_{\widehat{T}}^{T}$ of SCP \eqref{valuefunction1}. It is clear that $W_2$ has the same conclusion. Hence, $W_1=W_2=W$ on ${\cal X}_{\widehat{T}}^{T}$. Moreover, since the fixed point of ${\cal T}$ is unique, we know that $W_1=W_2=\widetilde{W}$ on ${\cal X}^{\widehat{T}}_{\widehat{T}-\epsilon}$. Therefore, $W_1=W_2$  on ${\cal X}_{\widehat{T}-\epsilon}^{T}$ and the uniqueness is obvious.

Noting the fact that $\epsilon$ is independent of $\widehat{T}$, by the bootstrap method, we can prove the conclusion on the whole interval $[0,T]$.
\end{proof}

\subsection{Verification theorem}
In this subsection, we give a verification theorem, which shows the relationship between the stochastic control problem and its associated Hamilton-Jacobi-Bellman-Isaacs  equation.

\begin{theorem} \label{verification theorem}
Suppose that Assumption B1, B2, B3 and A1, A2(N), A3(N) with $q=2$ are satisfied. Moreover, there exist functions $w\in C^{1,2}({\cal X}^T)\cap C(\overline{{\cal X}^T})$, and $U^*,\,V^*\in C(\overline{{\cal X}^T})$, satisfying the following partial differential equations (PDEs)
\begin{equation}\label{pde1}
 \left\{
 \begin{array}{l}
    -{\cal L}^{U^*,V^*}w
    =f^w(\cdot,\cdot,w,D_xw\cdot\beta(\cdot,\cdot,U^*,V^*),U^*,V^*)
    \;\;\mbox{on}\;\;{\cal X}^T;
    \vspace{2mm}\\
    -{\cal L}^{u,V^*} w
    \geq f^w(\cdot,\cdot,w,D_xw\cdot\beta(\cdot,\cdot,u,V^*),u,V^*),\;\;
    \forall\;u\in {\cal U},\;\;\mbox{on}\;\;{\cal X}^T;
    \vspace{2mm}\\
    -{\cal L}^{U^*,v} w
    \leq f^w(\cdot,\cdot,w,D_xw\cdot\beta(\cdot,\cdot,U^*,v),U^*,v),\;\;
    \forall\;v\in{\cal V},\;\;\mbox{on}\;\;{\cal X}^T;
    \vspace{2mm}\\
    w=\Psi\;\;\mbox{on}\;\;\partial{\cal X}^T,
 \end{array}
 \right.
\end{equation}
where the function $f^w$ is defined as in \eqref{newfunction1} and
$$
   D_x w\cdot\beta(t,x,u,v)
   \triangleq\left(\sum\limits_{i=1}^n\partial_{x_i}w(t,x)\beta_{i1}(t,x,u,v),\cdots,
   \sum\limits_{i=1}^n\partial_{x_i}w(t,x)\beta_{im}(t,x,u,v)\right)^{\rm T}
$$
and
\begin{equation*}
   {\cal L}^{u,v} w
   \triangleq \partial_t w
   +{1\over 2}\sum\limits_{i,j=1}^n\sum\limits_{k=1}^m
   \beta_{ik}(\cdot,\cdot,u,v)\beta_{jk}(\cdot,\cdot,u,v)\partial_{x_ix_j}w
   +\sum\limits_{i=1}^n\alpha_i(\cdot,\cdot,u,v)\partial_{x_i}w\,.
\end{equation*}

Moreover, suppose that $w$ and $U^*,V^*$ satisfy the following conditions:
\begin{enumerate}
\item there exists a constant $K_3$ such that $|w(t,x)|\leq \min(K_3b_1(t,x),b_4(x))$, $\forall(t,x)\in \overline{{\cal X}^T}$;

\item for any $\mathbf{u}\in{\cal U}(t,x)$ and $\mathbf{v}\in{\cal V}(t,x)$, SDE \eqref{stateequation1} has a unique strong solution if the pair of control processes is $(\mathbf{u}^{*;\mathbf{v}},\mathbf{v})$ or $(\mathbf{u},\mathbf{v}^{*;\mathbf{u}})$ with $\mathbf{u}^{*;\mathbf{v}}=U^*(\cdot,X^{t,x;\mathbf{u}^{*;\mathbf{v}},\mathbf{v}})\in {\cal U}(t,x)$ and $\mathbf{v}^{*;\mathbf{u}}=V^*(\cdot,X^{t,x;\mathbf{u},\mathbf{v}^{*;\mathbf{u}}})\in{\cal V}(t,x)$.
\end{enumerate}
Then $W=w$ is the unique value function of SCP \eqref{valuefunction1} in space ${\cal B}_{(0,T)}$, and $(\mathbf{u}^*,\mathbf{v}^*)$ is the optimal control strategy with  $\mathbf{u}^*=U^*(\cdot,X^{t,x;\mathbf{u}^*,\mathbf{v}^*})$ and $\mathbf{v}^*=V^*(\cdot,X^{t,x;\mathbf{u}^*,\mathbf{v}^*})$.
\end{theorem}

\begin{remark}\label{remark of main result2}
If the  generator  $f$ is independent of $y,z$, i.e.,
$$
   Y^{t,x;\mathbf{u},\mathbf{v}}_t
   =\mathbb{E}\left[\Psi\left(T\wedge\tau,X_{T\wedge\tau}^{t,x; \mathbf{u}, \mathbf{v}}\right)
   +\int_t^{T\wedge\tau} f\left(r,X_r^{t,x;\mathbf{u},\mathbf{v}},
   W\left(r,X_r^{t,x;\mathbf{u},\mathbf{v}}\right),\mathbf{u}_r,\mathbf{v}_r\right)\,dr\right].
$$
Then $q=2$ in Assumption A2(N) and A3(N) can be weaken into $q=1$. Its proof is similar to the following one.
\end{remark}

\begin{proof} We divide the proof into two steps.
\vspace{1mm}

\textbf{Step 1.} We prove the conclusion on the domain ${\cal X}_{(N-1)T/N}^{T}$. Firstly, we consider the standard SCP \eqref{valuefunction2} with $J=w$ on ${\cal X}_{(N-1)T/N}^{T}$ and prove that $w=W^w$ is the unique value function of SCP \eqref{valuefunction2}.

For any $(t,x)\in {\cal X}_{(N-1)T/N}^{T}$ and $\mathbf{u}\in{\cal U}(t,x)$, denoting $\widehat{Y}=w(\cdot,X^{t,x;\mathbf{u},\mathbf{v}^{*;\mathbf{u}}})$ and applying It\^o's formula to $\widehat{Y}$ yield
$$
   d\widehat{Y}_s=-\widehat{f}_s\,ds+\sum_{j=1}^m\widehat{Z}_{j,\,s}\,dB^j_s\,,
$$
where
$$
   \widehat{f}=-{\cal L}^{\mathbf{u},\mathbf{v}^{*;\mathbf{u}}} w\left(\cdot,X^{t,x;\mathbf{u},\mathbf{v}^{*;\mathbf{u}}}\right),\;\;
   \widehat{Z}_{j}=\sum_{i=1}^n \partial_{x_i}w\left(\cdot,X^{t,x;\mathbf{u},\mathbf{v}^{*;\mathbf{u}}}\right)
   \beta_{ij}\left(\cdot,
   X^{t,x;\mathbf{u},\mathbf{v}^{*;\mathbf{u}}},\mathbf{u},\mathbf{v}^{*;\mathbf{u}}\right).
$$

Next, we compare $\widehat{Y}$ and $Y^{w;t,x;\mathbf{u},\mathbf{v}^{*;\mathbf{u}}}$ in \eqref{payoff2} via the comparison theorem for BSDEs. Unfortunately, we don't know whether $\widehat{Z}$ belongs to $\mathbb{L}^2(t,T)$. So, we need some localization method to finish the comparison. Define
$$
   \widehat{\tau}_k=T\wedge\tau\wedge
   \inf\bigg(s\geq t:\Big|\sum\limits_{j=1}^m\int_t^s \widehat{Z}^k_{j,\,r}\,dB^j_r\Big|\leq k\bigg),\quad k=1,2,\cdots.
$$
Then $\widehat{Y}$ satisfies the following BSDE:
\begin{equation*}
   \widehat{Y}_sI_{\{t\leq s\leq \widehat{\tau}_k\}}
   +\widehat{Y}_{\widehat{\tau}_k}I_{\{\widehat{\tau}_k< s\leq T\}}
   =w\left(\widehat{\tau}_k,X_{\widehat{\tau}_k}^{t,x;\mathbf{u},\mathbf{v}^{*;\mathbf{u}}}\right)
   +\int_{s}^{T}I_{\{t\leq r\leq\widehat{\tau}_k\}}
   \bigg(\widehat{f}_r\,dr-\sum\limits_{j=1}^m\widehat{Z}_{j,\,r}\,dB^j_r\bigg)
\end{equation*}
for any $t\leq s\leq  T,\,k=1,2,\cdots$. Consider the following BSDE:
\begin{eqnarray*}
   \widetilde{Y}^k_s&=&
   w\left(\widehat{\tau}_k,
   X_{\widehat{\tau}_k}^{t,x;\mathbf{u},\mathbf{v}^{*;\mathbf{u}}}\right)
   +\int_s^T f^w\left(r,X_r^{t,x;\mathbf{u},\mathbf{v}^{*;\mathbf{u}}},
   \widetilde{Y}^k_r,
   \widetilde{Z}^k_r,\mathbf{u}_r,\mathbf{v}^{*;\mathbf{u}}_r\right)
   I_{\{t\leq r\leq\widehat{\tau}_k\}}\,dr
   \\[2mm]
   &&-\sum\limits_{j=1}^m\int_s^T \widetilde{Z}^k_{j,\,r}\,dB^j_r.
\end{eqnarray*}
From $|w(t,x)|\leq b_4(x)$ and Assumption A2(N), we know that  $w(\widehat{\tau}_k,X_{\widehat{\tau}_k}^{t,x;\mathbf{u},\mathbf{v}^{*;\mathbf{u}}})\in L_T^2$. By the same argument as in the proof of Theorem 1, we can deduce that the above BSDE has a unique solution $(\widetilde{Y}^k,\widetilde{Z}^k)\in \mathbb{S}^2(t,T)\times\mathbb{L}^2(t,T)$. Moreover, $|w(t,x)|\leq b_4(x)$ and Assumption A2(N) imply $(\widehat{Y}I_{\{t\leq\cdot\leq \widehat{\tau}_k\}} +\widehat{Y}_{\widehat{\tau}_k}I_{\{\widehat{\tau}_k<\cdot\leq T\}}, \widehat{Z}I_{\{t\leq \cdot\leq \widehat{\tau}_k\}})\in\mathbb{S}^2(t,T)\times\mathbb{L}^2(t,T)$.

The first inequality in PDE \eqref{pde1} yields
$$
   \widehat{f}_r
   =-{\cal L}^{\mathbf{u}_r,\mathbf{v}_r^{*;\mathbf{u}}} w\left(r,X_r^{t,x;\mathbf{u},\mathbf{v}^{*;\mathbf{u}}}\right)
   \geq f^w\left(r,X_r^{t,x;\mathbf{u},\mathbf{v}^{*;\mathbf{u}}},
   \widehat{Y}_r,\widehat{Z}_r,\mathbf{u}_r,\mathbf{v}^{*;\mathbf{u}}_r\right),
   \;\;\forall\;t\leq s\leq T\wedge\widehat{\tau}.
$$
Hence, by the comparison principle for BSDEs, we obtain, for any $s\in[t,T]$ and $k=1,2,\cdots$,
\begin{equation}\label{comparison1}
   \widehat{Y}_sI_{\{t\leq s\leq\widehat{\tau}_k\}}
   \geq \widetilde{Y}^k_sI_{\{t\leq s\leq\widehat{\tau}_k\}},\;\;\mbox{a.s.}
\end{equation}

On the other hand, BSDE \eqref{payoff2} gives for any $t\leq s\leq T$,
\begin{eqnarray*}
   &&\hspace{-0.8cm}Y_sI_{\{t\leq s\leq T\wedge\widehat{\tau}\}}
   +Y_{T\wedge\widehat{\tau}}I_{\{T\wedge\widehat{\tau}<s\leq T\}}
   =\int_s^Tf^w\left(r,X_r,Y_r,Z_r,\mathbf{u}_r,\mathbf{v}^{*;\mathbf{u}}_r\right)
   I_{\{{t\leq r\leq T\wedge\widehat{\tau}}\}}\,dr
   \\
   &&\hspace{5cm}-\sum\limits_{j=1}^m\int_s^T Z_{j,\,r}
   I_{\{{t\leq r\leq T\wedge\widehat{\tau}}\}}\,dB^j_r
   +\Psi\left(T\wedge\widehat{\tau},X_{T\wedge\widehat{\tau}}\right),
\end{eqnarray*}
where $X,Y$ and $Z$ denote $X^{t,x;\mathbf{u},\mathbf{v}^{*;\mathbf{u}}},  Y^{w;t,x;\mathbf{u},\mathbf{v}^{*;\mathbf{u}}}$ and $Z^{w;t,x;\mathbf{u},\mathbf{v}^{*;\mathbf{u}}}$, respectively. Applying the standard estimate for BSDEs, we have
\begin{eqnarray}\label{estimate1inthoerem2}
  \left\|\left(Y-\widetilde{Y}^k\right)
  I_{\{t\leq \cdot\leq T\wedge\widehat{\tau}\}}\right\|_{\mathbb{S}^2(t,T)}
  &\leq& C\left(\,\left\|w(\widehat{\tau}_k,X_{\widehat{\tau}_k})
  -\Psi\left(T\wedge\widehat{\tau},X_{T\wedge\widehat{\tau}}\right)\right\|_{L^2_T}\right.
  \nonumber\\
  &&\left.+\left\|f^w\left(\cdot,X,Y,Z,
  \mathbf{u},\mathbf{v}^{*;\mathbf{u}}\right)
  I_{\{{\widehat{\tau}_k\leq \cdot\leq T\wedge\widehat{\tau}}\}}\right\|_{\mathbb{L}^2(t,T)}\,\right).
\end{eqnarray}

With the fact that $\widehat{\tau}_k\rightarrow\widehat{\tau}$ a.s. as $k\rightarrow+\infty$, the boundary condition in PDE~\eqref{pde1} and the continuities of $X$ and $w$ imply that $w(\widehat{\tau}_k,X_{\widehat{\tau}_k})
-\Psi(T\wedge\widehat{\tau},X_{T\wedge\widehat{\tau}})\rightarrow0$ a.s. as $k\rightarrow+\infty$. Again, by $|w(t,x)|\leq b_4(x)$, Assumption B3 and Assumption A2(N), we get
\begin{eqnarray*}
   \left\|w\left(\widehat{\tau}_k,X_{\widehat{\tau}_k}\right)
   -\Psi\left(T\wedge\widehat{\tau},X_{T\wedge\widehat{\tau}}\right)\right\|_{L^2_T}
   \leq
   \bigg\|\sup\limits_{t\leq s\leq T\wedge\widehat{\tau}}2b_4(X_s)\bigg\|_{L^2_T}
   \leq2\|b_4(X)\|_{\mathbb{S}^2(t,T)}<+\infty\,.
\end{eqnarray*}
Then, by the dominated convergence theorem, we obtain
$$
   \left\|w\left(\widehat{\tau}_k,X_{\widehat{\tau}_k}\right)
   -\Psi\left(T\wedge\widehat{\tau},X_{T\wedge\widehat{\tau}}\right)\right\|_{L^2_T}
   \rightarrow0,\;\;\mbox{as }k\rightarrow+\infty.
$$
Moreover, it is clear that $I_{\{{\widehat{\tau}^k\leq\cdot\leq T\wedge\widehat{\tau}}\}}\rightarrow0$ a.e. on $\Omega\times[0,T]$ as $k\rightarrow+\infty$. The proof of Theorem 1 implies
\begin{equation*}
   \left\|f^w\left(\cdot,X,Y,Z,\mathbf{u},
   \mathbf{v}^{*;\mathbf{u}}\right)\right\|_{\mathbb{L}^2(t,T)}
   \leq \left\|f^w\left(\cdot,X,0,0,\mathbf{u},\mathbf{v}^{*;\mathbf{u}}\right)
   +C(|Y|+|Z|)\right\|_{\mathbb{L}^2(t,T)}
   <\infty.
\end{equation*}
Again, by dominated convergence theorem, we deduce
$$
   \left\|f^w\left(\cdot,X,Y,Z,\mathbf{u},\mathbf{v}^{*;\mathbf{u}}\right)
   I_{\{{\widehat{\tau}_k\leq \cdot\leq T\wedge\widehat{\tau}}\}}
   \right\|_{\mathbb{L}^2(t,T)}
   \rightarrow0,\;\;\mbox{as }k\rightarrow+\infty.
$$
Hence, from~\eqref{estimate1inthoerem2}, we derive
$$
   \left\|\left(Y-\widetilde{Y}^k\right)
   I_{\{t\leq \cdot\leq T\wedge\widehat{\tau}\}}\right\|_{\mathbb{S}^2(t,T)}
   \rightarrow0,\;\;\mbox{as}\;\;k\rightarrow+\infty.
$$
Recalling~\eqref{comparison1}, we conclude that $Y_s^{w;t,x;\mathbf{u},\mathbf{v}^{*;\mathbf{u}}}\leq \widehat{Y}_s$ a.s. for any $t\leq s\leq\widehat{\tau}$. Particularly,
$$
   Y_t^{w;t,x;\mathbf{u},\mathbf{v}^{*;\mathbf{u}}}
   \leq\widehat{Y}_t=w(t,x),\quad \forall\;\mathbf{u}\in{\cal U}(t,x).
$$
Repeating the similar argument as the above and via the second inequality in PDE \eqref{pde1}, we get
$$
   w(t,x)\leq Y_t^{w;t,x;\mathbf{u}^{*;\mathbf{v}},\mathbf{v}},\quad
   \forall\;\mathbf{v}\in{\cal V}(t,x).
$$
Therefore, we obtain
\begin{equation}\label{saddlepoint}
   W^w(t,x)\leq
   \sup\limits_{\mathbf{u}\in{\cal U}(t,x)}
   Y_t^{w;t,x;\mathbf{u},\mathbf{v}^{*;\mathbf{u}}}\leq w(t,x)
   \leq\inf\limits_{\mathbf{v}\in{\cal V}(t,x)}
   Y_t^{w;t,x;\mathbf{u}^{*;\mathbf{v}},\mathbf{v}}
   \leq W^w(t,x).
\end{equation}
Until now, we have proved that $w=W^w$ on ${\cal X}_{(N-1)T/N}^{T}$, which is the unique value function of SCP \eqref{valuefunction2} with $J=w$. 
Moreover, the first equation in PDE \eqref{pde1} provides that $w(t,x)=Y_t^{w;t,x;\mathbf{u}^*,\mathbf{v}^*}$ and $(\mathbf{u}^*,\mathbf{v}^*)$ is the optimal control strategy.

Noting that in the case where the value function $w$ is the super index of $W^w$, $w$ is also a value function of the original SCP \eqref{valuefunction1}. Thanks to Theorem 1, we know that the value function of SCP \eqref{valuefunction1} in ${\cal B}_{((N-1)T/N,T)}$ is unique. Since $|w(t,x)|\leq K_3b_1(t,x)$ for any $(t,x)\in\overline{\mathcal{X}^T}$, then $w\in{\cal B}_{((N-1)T/N,T)}$. Therefore, $w$ is the unique value function of SCP \eqref{valuefunction2} in ${\cal B}_{((N-1)T/N,T)}$.
\vspace{1mm}

\textbf{Step 2.} We extend the conclusion to the whole domain ${\cal X}^T$. Assume that we have proved that $w$ and $(\mathbf{u}^*,\mathbf{v}^*)$ are respectively the unique value function in space ${\cal B}_{(iT/N,T)}$ and the optimal control strategy of SCP \eqref{valuefunction1} (Problem {\bf O}, for short) on ${\cal X}^T_{iT/N}$.
Consider a new auxiliary SCP \eqref{valuefunction1} (Problem {\bf A}, for short) on ${\cal X}_{(i-1)T/N}^{iT/N}$, where the terminal time $T$ and the terminal value $\Psi(t,x)$ are replaced by $iT/N$ and $\widehat{\Psi}(t,x)=\Psi(t,x)I_{\{t<iT/N\}}+w(iT/N,x)$, respectively. Noting that $|w(t,x)|\leq b_4(x)$ for any $(t,x)\in\overline{\mathcal{X}^T}$, one can apply Theorem 1 to Problem {\bf A}, 
hence, we know that Problem {\bf A} has a unique value function in space ${\cal B}_{((i-1)T/N,iT/N)}$. Moreover, repeating the argument in Step 1, we deduce that $w$ and $(\mathbf{u}^*,\mathbf{v}^*)$ are the value function and the optimal control strategy of Problem {\bf A}, respectively.
\vspace{1mm}

We first prove that $w$ is the value function in space ${\cal B}_{((i-1)T/N,T)}$ of Problem {\bf O} on ${\cal X}_{(i-1)T/N}^{T}$. From the fact $|w|\leq K_3b_1$ on ${\cal X}^{T}$,  we know $w\in{\cal B}_{((i-1)T/N,T)}$.
As $w$ is the value function of Problem {\bf A}, $w$ is also the value function of the standard SCP \eqref{valuefunction2} (Problem {\bf S}, for short) with $J=w$, the terminal time $iT/N$ and the terminal value $\widehat{\Psi}$.
Moreover, since $w$ is the value function of Problem {\bf O} on ${\cal X}^T_{iT/N}$, $w$ is also the value function of Problem {\bf S} with $J=w$ on ${\cal X}^{T}_{iT/N}$.
Then, by Assumption A1 and the DPP for the standard SCP \eqref{valuefunction2}, we know that $w$ is the value function  of Problem {\bf S} on ${\cal X}^{T}_{(i-1)T/N}$ with $J=w$. Hence, $w$ is also the value function in ${\cal B}_{((i-1)T/N,T)}$ of Problem {\bf O} on ${\cal X}_{(i-1)T/N}^{T}$.
\vspace{1mm}

Next, we establish the uniqueness of the value function in space ${\cal B}_{((i-1)T/N,T)}$ of Problem {\bf O} on ${\cal X}_{(i-1)T/N}^{T}$.
Assume $W_1,W_2\in{\cal B}_{((i-1)T/N,T)}$ are the value functions of Problem {\bf O} on ${\cal X}_{(i-1)T/N}^{T}$. Then $W_1$ is also the value function of Problem {\bf S} with $J=W_1$. Thank to Assumption A1 and the DPP for the standard SCP \eqref{valuefunction2}, we know that $W_1$ is the unique value function of Problem {\bf S} on ${\cal X}^{T}_{iT/N}$ with $J=W_1$, which is also the unique value function in space ${\cal B}_{(iT/N,T)}$ of Problem {\bf O} on ${\cal X}^{T}_{iT/N}$. Hence, $W_1=w$ on ${\cal X}^{T}_{iT/N}$.
Moreover, repeating the same argument, we know that $W_1$ is the unique value function in space ${\cal B}_{((i-1)T/N,iT/N)}$ of Problem {\bf A}, which provides $W_1=w$ on ${\cal X}_{(i-1)T/N}^{iT/N}$. It is clear that $W_2$ has the same properties. Therefore, we deduce $W_1=W_2=w$ on ${\cal X}_{(i-1)T/N}^{T}$.
Still by Assumption A1 and the DPP for the standard SCP \eqref{valuefunction2}, we know $w(t,x)=Y_t^{w;t,x;\mathbf{u}^*,\mathbf{v}^*}$ on ${\cal X}_{(i-1)T/N}^{T}$, which means that $(\mathbf{u}^*,\mathbf{v}^*)$ is the optimal control strategy on ${\cal X}_{(i-1)T/N}^{T}$.

Until now, we have proved the conclusion on $[(i-1)T/N,T]$.  Using the bootstrap method, we can extend the conclusion to the whole interval $[0,T]$.
\end{proof}

\begin{remark} From \eqref{saddlepoint}, we know that the following relation holds:
$$
   w(t,x)
   =\sup\limits_{\mathbf{u}\in{\cal U}(t,x)}\inf\limits_{\mathbf{v}\in{\cal V}(t,x)}Y^{t,x;\mathbf{u},\mathbf{v}}_t
   =\inf\limits_{\mathbf{v}\in{\cal V}(t,x)}\sup\limits_{\mathbf{u}\in{\cal U}(t,x)}Y^{t,x;\mathbf{u},\mathbf{v}}_t.
$$
\end{remark}

\section{Applications}
In this section, we apply the above results to solve two typical examples in finance.

\subsection{Example 1}
This problem comes from \cite{Maenhout}, in which the author gave heuristically the model and its associated HJBI equation rather than the accurate model and description. We will provide the precise model and the mathematical proofs.

Suppose that the price of the risk-free asset is governed by
\begin{equation}\label{riskfreeasset1}
   S^{t,S_0}_{0,s}=S_0+\int_t^s\mu_0S^{t,S_0}_{0,r}\,dr,
\end{equation}
where $\mu_0$ is the constant risk-free interest rate. The price of the risky asset on the probability space $(\Omega,{\cal F},\mathbb{F},\mathbb{P})$ is governed by
\begin{equation*}
   S^{t,S;\pmb{\phi}}_s=S+\int_{t}^s (\mu_1+\sigma\pmb{\phi}_r) S^{t,S;\pmb{\phi}}_rdr
   +\int_{t}^s \sigma S_r^{t,S;\pmb{\phi}}dB_r,
\end{equation*}
where $\mu_1$ and $\sigma$ are respectively constant reference expectation return rate and positive constant volatility rate. The ambiguity parameter $\pmb{\phi}$
is a progressively measurable process with respect to $\{{\cal F}_s^t\}_{t\leq s\leq T}$ taking values in $\mathbb{R}$.

The investor invests in the risky asset according to the proportion $\pmb{\pi}$ of the wealth, consumes at the rate of $\pmb{c}$ and invests the remain wealth into the risk-free asset.  So, the wealth $X$ is expressed as following,
\begin{equation}\label{wealthsde1}
   X_s^{t,x;\pmb{\pi},\pmb{c},\pmb{\phi}}
   =x+\int_t^s\Big\{\Big[\,\mu_0+(\mu_1-\mu_0+\sigma\pmb{\phi}_r)\pmb{\pi}_r\,\Big]
   \,X_r^{t,x;\pmb{\pi},\pmb{c},\pmb{\phi}}-\pmb{c}_r\,\Big\}\,dr+\int_t^s\sigma
   \pmb{\pi}_rX_r^{t,x;\pmb{\pi},\pmb{c},\pmb{\phi}}dB_r,
\end{equation}
where $0\leq t\leq s\leq T$, $x\in(0,+\infty)$ indicates the initial value, $\pmb{\pi}$ and $\pmb{c}$ are progressively measurable processes with respect to $\{{\cal F}_s^t\}_{t\leq s\leq T}$ taking values in $\mathbb{R}$ and $(0,+\infty)$, respectively.

The investor chooses the investment and consumption strategy $(\pmb{\pi},\,\pmb{c})$ to maximize the robust expectation utility $\inf\{{\cal J}_1(t,x;\pmb{\pi},\pmb{c},\pmb{\phi}):\pmb{\phi}\in{\Phi}_1(t,x)\}$ with
\begin{eqnarray*}
   {\cal J}_1(t,x;\pmb{\pi},\pmb{c},\pmb{\phi})&\triangleq&
   \mathbb{E}\bigg\{\,\int_t^T e^{-\delta (T-s)}\,
   \Big[\,{\pmb{c}_s^{1-\gamma} \over 1-\gamma}
   +{(1-\gamma)\sigma^2\pmb{\phi}_s^2\over 2\theta}\,W\left(t,X^{t,x;\pmb{\pi},\pmb{c},\pmb{\phi}}\right)\,\Big]\,ds
   \\[2mm]
   &&\qquad+e^{-\delta (T-t)}
   {\left(X^{t,x;\pmb{\pi},\pmb{c},\pmb{\phi}}_T\right)^{1-\gamma}\over1-\gamma}\,\bigg\},
\end{eqnarray*}
where $\delta>0$ is the constant discounted rate. The first term in the bracket represents the utility from the consumption. The constant $\gamma$ is the relative risk aversion coefficient satisfying $\gamma>0$ and $\gamma\neq1$. The second term represents the entropy penalty, where the constant $\theta>0$ measures the strength of the preference for robustness. $W$ is defined as
\begin{equation}\label{valuefunction3}
   W(t,x)=\sup\limits_{(\pmb{\pi},\,\pmb{c})\in{\Pi}_1(t,x)}
   \inf\limits_{\pmb{\phi}\in{\Phi}_1(t,x)}{\cal J}_1(t,x;\pmb{\pi},\pmb{c},\pmb{\phi}).
\end{equation}

In \cite{Maenhout}, the author didn't give the exact optimal control problem nor the admissible sets ${\Phi}_1(t,x)$ and ${\Pi}_1(t,x)$. Now, we take the admissible sets as follows:
\begin{eqnarray*}\label{addmissble2}
   {\Phi}_1(t,x)&\!\!\!=\!\!\!&\{\pmb{\phi}:|\pmb{\phi}|\leq K_4\},\\[2mm]
   {\Pi}_1(t,x)&\!\!\!=\!\!\!&\{(\pmb{\pi},\pmb{c}):|\pmb{\pi}|\leq K_4,X^{t,x;\pmb{\pi},\pmb{c},\pmb{\phi}}/K_4\leq \pmb{c}\leq K_4X^{t,x;\pmb{\pi},\pmb{c},\pmb{\phi}},X^{t,x;\pmb{\pi},\pmb{c},\pmb{\phi}}\geq0\},
\end{eqnarray*}
with a large enough constant $K_4$.
\begin{remark}
Unlike in the classical case, we introduce the constant $K_4$ and suppose it is large enough in order to ensure BSDE \eqref{payoff1} admits $\mathbb{L}^2-$solution. Moreover, the optimal strategy falls in this admissible set.
\end{remark}
\begin{theorem}
Suppose that there exist functions $w\in C^{1,2}([0,T]\times (0,+\infty))$ and $\pi^*,c^*,\phi^*\in C^{0,1}([0,T]\times [0,+\infty))$, which satisfy the following PDEs:
\begin{equation*}
 \left\{
 \begin{array}{l}
   -{\cal L}_1^{\pi^*(t,x),c^*(t,x),\phi^*(t,x)}w(t,x)=f_1(w(t,x),c^*(t,x),\phi^*(t,x))\,;\vspace{2mm}\\
   -{\cal L}_1^{\pi,c,\phi^*(t,x)}w(t,x)\geq f_1(w(t,x),c,\phi^*(t,x)),\;\;\forall\;|\pi|\leq K_4,\,x/K_4\leq c\leq K_4x\,;\vspace{2mm}\\
   -{\cal L}_1^{\pi^*(t,x),c^*(t,x),\phi}w(t,x)\leq f_1(w(t,x),c^*(t,x),\phi),\;\;\forall\;|\phi|\leq K_4\,;\vspace{2mm}\\
   w(T,x)=x^{1-\gamma}/(1-\gamma)\,,
 \end{array}
 \right.
\end{equation*}
where the function $f_1$ is defined as
$$
   f_1(w,c,\phi)\triangleq -\delta w
   +{c^{1-\gamma}\over1-\gamma}+{(1-\gamma)\sigma^2\phi^2\over 2\theta}\,w
$$
and the differential operator ${\cal L}_1^{\pi,c,\phi}$ is defined as
\begin{equation*}
   {\cal L}_1^{\pi,c,\phi} w\triangleq
   \partial_t w+{1\over 2}\sigma^2\pi^2x^2\partial_{xx}w
   +\Big\{\,\Big[\,\mu_0+(\mu_1-\mu_0+\sigma\phi)\pi\,\Big]\,x-c\,\Big\}\partial_{x}w.
\end{equation*}
Moreover, assume that $w$ and $\pi^*,c^*,\phi^*$ satisfy the following conditions:
\begin{enumerate}
 \item there exists a constant $K_5$ such that
 $|w|\leq K_5x^{1-\gamma}$ on $[0,T]\times(0,+\infty)$;\smallskip

 \item $|\pi^*|\leq K_4$, $|\phi^*|\leq K_4$ and $x/K_4\leq c^*\leq K_4x$ on $[0,T]\times(0,+\infty)$.
\end{enumerate}
Then $W=w$ is the unique value function in ${\cal B}_{(0,T)}$ of SCP \eqref{valuefunction3} with $b_1(t,x)=x^{1-\gamma}$, and $(\pmb{\pi}^*,\pmb{c}^*,\pmb{\phi}^*)\triangleq
(\pi^*,c^*,\phi^*)(\cdot,X^{t,x;\pmb{\pi}^*,\pmb{c}^*,\pmb{\phi}^*})$ is the optimal control strategy.
\end{theorem}

\begin{proof}
It is clear that SCP \eqref{valuefunction3} can be described as the form of SCP \eqref{valuefunction1} with
\begin{eqnarray*}
   &&\mathbf{u}=(\pmb{\pi},\pmb{c}),\;\;
   \mathbf{v}=\pmb{\phi},\;\;{\cal X}=(0,+\infty),\;\;{\cal U}=[\,-K_4,K_4\,]\times(0,+\infty),\;\;
   {\cal V}=[\,-K_4,K_4\,],
   \\[2mm]
   &&\alpha(t,x,\pi,c,\phi)=\Big[\,\mu_0+(\mu_1-\mu_0+\sigma\phi)\pi\,\Big]\,x-c,\quad
   \beta(t,x,\pi,c,\phi)=\sigma\pi x,
   \\[2mm]
   &&\Psi={x^{1-\gamma}\over1-\gamma},\;\;f(t,x,y,z,w,\pi,c,\phi)=-\delta y
   +{c^{1-\gamma}\over 1-\gamma}+{(1-\gamma)\sigma^2\phi^2\over 2\theta}\,w.
\end{eqnarray*}

It is not difficult to verify that Assumption B1-B3 are satisfied with
$$
   K_1=\delta,\quad
   b_2(\pi,\phi)={|1-\gamma|\sigma^2\phi^2\over 2\theta},\quad
   b_3(\pi,c,\phi)={c^{1-\gamma}\over |1-\gamma|},\quad
   b_4(x)=K_5x^{1-\gamma}.
$$

For any $(t,x)\in[0,T]\times (0,+\infty)$, $(\pmb{\pi},\pmb{c})\in \Pi_1(t,x)$ and $\pmb{\phi}\in \Phi_1(t,x)$, one can easily check that the coefficients of SDE \eqref{wealthsde1} are Lipschitz continuous and linear growth with respect to $x$. Then, SDE \eqref{wealthsde1} has a unique strong solution $X^{t,x;\pmb{\pi},\pmb{c},\pmb{\phi}}$ (refer to P82 in \cite{Krylov}). Moreover, it is not difficult to check that the admissible sets $\Pi_1(t,x)$ and $\Phi_1(t,x)$  satisfy the ``switching condition''. So, Assumption A1 is satisfied.

A simple calculation yields that $\widetilde{X}=(X^{t,x;\pmb{\pi},\pmb{c},\pmb{\phi}})^{1-\gamma}$ satisfies $\widetilde{X}_t=x^{1-\gamma}$ and
\begin{equation*}
   d\widetilde{X}_s
   =(1-\gamma)\left[\,\mu_0+(\mu_1-\mu_0+\sigma\pmb{\phi}_s)\pmb{\pi}_s
   -{\pmb{c}_s\over X^{t,x;\pmb{\pi},\pmb{c},\pmb{\phi}}_s}
   -{\gamma\over 2}\sigma^2\pmb{\pi}^2_s\,\right]\,\widetilde{X}_s\,ds
   +(1-\gamma)\sigma \pmb{\pi}_s\widetilde{X}_sdB_s.
\end{equation*}
Applying the theory for SDEs  (refer to P86 in \cite{Krylov}), we get the following estimate:
\begin{equation*}
   \left\|\left(X^{t,x;\pmb{\pi},\pmb{c},\pmb{\phi}}\right)^{1-\gamma}\right\|_{\mathbb{S}^2(t,T)}
   =\left\|\widetilde{X}\right\|_{\mathbb{S}^2(t,T)}
   \leq Cx^{1-\gamma},
\end{equation*}
where the constant $C$ is independent of $t,x,\pmb{\pi},\pmb{c},\pmb{\phi}$. Then, with the expressions of $b_1,b_2,b_3,b_4$, we have
\begin{eqnarray*}
   &&\left\|b_1\left(X^{t,x;\pmb{\pi},\pmb{c},\pmb{\phi}}\right)
   b_2(\pmb{\pi},\pmb{c},\pmb{\phi})\right\|_{\mathbb{L}^2(t,T)}
   +\|b_3(\pmb{\pi},\pmb{c},\pmb{\phi})\|_{\mathbb{L}^2(t,T)}
   +\left\|b_4\left(X^{t,x;\pmb{\pi},\pmb{c},\pmb{\phi}}\right)\right\|_{\mathbb{S}^2(t,T)}
   \\[2mm]
   &\leq &C\left(\left\|\left(X^{t,x;\pmb{\pi},\pmb{c},\pmb{\phi}}\right)^{1-\gamma}\right\|_{\mathbb{L}^2(t,T)}
   +\left\|\pmb{c}^{1-\gamma}\right\|_{\mathbb{L}^2(t,T)}
   +\left\|\left(X^{t,x;\pmb{\pi},\pmb{c},\pmb{\phi}}\right)^{1-\gamma}\right\|_{\mathbb{S}^2(t,T)}\right)
   \leq  \widetilde{C}x^{1-\gamma},
\end{eqnarray*}
where the constant $\widetilde{C}$ is also independent of $t,x,\pmb{\pi},\pmb{c},\pmb{\phi}$. Hence, we have testified Assumption A2(1).

Moreover, it is not difficult to deduce
$$
   \left\|b_1\left(X^{t,x;\pmb{\pi},\pmb{c},\pmb{\phi}}\right)
   b_2(\pmb{\pi},\pmb{c},\pmb{\phi})\right\|_{\mathbb{L}^2(t,t+\epsilon)}
   \leq \left\|b_1\left(X^{t,x;\pmb{\pi},\pmb{c},\pmb{\phi}}\right)
   b_2(\pmb{\pi},\pmb{c},\pmb{\phi})\right\|_{\mathbb{S}^2(t,T)})\sqrt{\epsilon}
   \leq \widetilde{C}\sqrt{\epsilon}x^{1-\gamma},
$$
where the constant $\widetilde{C}$ is independent of $t,x,\pmb{\pi},\pmb{c},\pmb{\phi}$. Until now, we have verified all conditions in Theorem 1. Therefore, we conclude that SCP \eqref{valuefunction3} has a unique value function $W$ in space ${\cal B}_{(0,T)}$.

Noting that  $|w(t,x)|\leq \min(K_5b_1(t,x),b_4(x))$ for any $(t,x)\in [0,T]\times(0,+\infty)$ and
$$
   f^w(t,x,w(t,x),\beta(t,x,u,v)\partial_xw(t,x),u,v)=f_1(w(t,x),c,\phi),
$$
we can achieve the desired result by means of Theorem 2 if we can prove Condition 2 in Theorem 2.

For any $\pmb{\phi}\in \Phi_1(t,x)$, the control process is $(\pmb{\pi}^{*;\pmb{\phi}},\pmb{c}^{*;\pmb{\phi}},\pmb{\phi})$ with $\pmb{\pi}^{*;\pmb{\phi}}=\pi^*(\cdot,
X^{t,x;\pmb{\pi}^{*;\pmb{\phi}},\pmb{c}^{*;\pmb{\phi}},\pmb{\phi}})$ and $\pmb{c}^{*;\pmb{\phi}}=c^*(\cdot,
X^{t,x;\pmb{\pi}^{*;\pmb{\phi}},\pmb{c}^{*;\pmb{\phi}},\pmb{\phi}})$, then SDE \eqref{wealthsde1} becomes
\begin{eqnarray*}
   X_s^{t,x;\pmb{\pi}^{*;\pmb{\phi}},\pmb{c}^{*;\pmb{\phi}},\pmb{\phi}}
   &=&x+\int_t^s\widetilde{\alpha}_r\left(
   X^{t,x;\pmb{\pi}^{*;\pmb{\phi}},\pmb{c}^{*;\pmb{\phi}},\pmb{\phi}}_r\right)\,dr
   +\int_t^s\widetilde{\beta}_r\left(
   X^{t,x;\pmb{\pi}^{*;\pmb{\phi}},\pmb{c}^{*;\pmb{\phi}},\pmb{\phi}}_r\right)dB_r,
\end{eqnarray*}
with the coefficients
$$
   \widetilde{\alpha}_r(x)=\Big\{\Big[\,\mu_0+(\mu_1-\mu_0+\sigma\pmb{\phi}_r)\pi^*(r,x)\,\Big]\,x
   -c^*(r,x)\,\Big\}\quad{\rm and}\quad
   \widetilde{\beta}_r(x)=\sigma\pi^*(r,x)\,.
$$

Since $\pi^*,c^*\in C^{0,1}([0,T]\times (0,+\infty))$, $|\pi^*(t,x)|\leq K_4$ and $x/K_4\leq c^*(t,x)\leq K_4x$, $\forall(t,x)\in[0,T]\times(0,+\infty)$ and $\forall t\in[0,T]$, $|\pmb{\phi}_t|\leq K_4$ a.s., we know that $\widetilde{\alpha}_r(x)$ and $\widetilde{\beta}_r(x)$ are locally Lipschitz continuous with respect to $x$, i.e., $\forall n\in\mathbb{N}^+$, there exists a constant $C_n$ such that
$$
   |\widetilde{\alpha}_r(x)-\widetilde{\alpha}_r(x^\prime)|
   +|\widetilde{\beta}_r(x)-\widetilde{\beta}_r(x^\prime)|
   \leq C_n|x-x^\prime|,\ \mbox{a.s.},\;\forall\;x,x^\prime\in[0,n],
   r\in[0,T].
$$
Also, $\widetilde{\alpha}_r(x)$ and $\widetilde{\beta}_r(x)$ are linear growth with respect to $x$, i.e., there exists a constant $C$ such that
$$
   |\widetilde{\alpha}_r(x)|+|\widetilde{\beta}_r(x)|
   \leq C(1+|x|),\ \mbox{a.s.},\;\forall\;(r,x)\in[0,T]\times(0,+\infty).
$$
By the theory for SDEs (refer to P57 in \cite{Mao}), we deduce that SDE \eqref{wealthsde1} has a unique strong solution $X^{t,x;\pmb{\pi}^{*;\pmb{\phi}},\pmb{c}^{*;\pmb{\phi}},\pmb{\phi}}$. Moreover, Condition 2 in Theorem 2 indicates $X^{t,x;\pmb{\pi}^{*;\pmb{\phi}},\pmb{c}^{*;\pmb{\phi}},\pmb{\phi}}/K_4
\leq\pmb{c}^{*;\pmb{\phi}}
\leq K_4X^{t,x;\pmb{\pi}^{*;\pmb{\phi}},\pmb{c}^{*;\pmb{\phi}},\pmb{\phi}}$ and $|\pmb{\pi}^{*;\pmb{\phi}}|\leq K_4$. That is, $(\pmb{\pi}^{*;\pmb{\phi}},\pmb{c}^{*;\pmb{\phi}})\in \Pi_1(t,x)$.

Repeating the similar argument, we can prove that for any $(\pmb{\pi},\pmb{c})\in \Pi_1(t,x)$, SDE~\eqref{wealthsde1} has a unique solution if the control process is $(\pmb{\pi},\pmb{c},\pmb{\phi}^{*;\pmb{\pi},\pmb{c}})$ and $\pmb{\phi}^{*;\pmb{\pi},\pmb{c}}\in\Phi_1(t,x)$ with $\pmb{\phi}^{*;\pmb{\pi},\pmb{c}}
=\phi^*(\cdot,X^{t,x;\pmb{\pi},\pmb{c},\pmb{\phi}^{*;\pmb{\pi},\pmb{c}}})$.

Hence, we have testified all assumptions in Theorem 2, and the desired result is obvious.
\end{proof}

According to the conclusion in \cite{Maenhout}, we give the explicit form of the value function to finish the example.

\begin{theorem}
The unique value function $W$ in the space ${\cal B}_{(0,T)}$ with $b_1(t,x)=x^{1-\gamma}$,  and the optimal strategy functions $\pi^*,\,c^*,\,\phi^*$ of SCP \eqref{valuefunction3} are as follows,
\begin{eqnarray*}
   W(t,x)=\left[\,g_1(t)\,\right]^\gamma{x^{1-\gamma}\over1-\gamma},\;\;
   \pi^*(t,x)={\mu_1-\mu_0\over\theta+\gamma\sigma^2},\;\;
   c^*(t,x)={x\over g_1(t)},\;\;
   \phi^*(t,x)={\theta(\mu_0-\mu_1)\over\sigma(\theta+\gamma\sigma^2)},
\end{eqnarray*}
\begin{equation*}
   g_1(t)=
   \left\{
   \begin{array}{l}
     {\displaystyle {1+(a_1-1)e^{a_1(t-T)}\over a_1},
     \mbox{if}\;a_1\neq0;}
     \vspace{2mm}\\
     {\displaystyle T+1-t,
     \qquad\qquad\;\;
     \mbox{if}\;a_1=0,}
   \end{array}
   \right.
   \mbox{with}\;a_1={\delta\over\gamma}
   -{1-\gamma\over\gamma}\left[\,\mu_0
   +{\gamma\sigma^2(\mu_1-\mu_0)^2\over 2(\theta+\gamma\sigma^2)}\,\right].
\end{equation*}
It is not difficult to verify that $w=W$ and $\pi^*,\,c^*,\,\phi^*$ satisfy all the assumptions in Theorem 3, provided $K_4,K_5$ are large enough. We omit its details.
\end{theorem}

\subsection{Example 2}

This problem comes from \cite{Escobar}, in which the authors neither proved the rationality of the stochastic control problem nor gave the explicit forms of the value function and the optimal strategy, they only guessed the expressions of the value function and the optimal strategy. What we will do is to give their explicit forms and mathematic proofs via Theorem 2. In order to focus on the mathematical essence of the problem, we omit the jump and the derivative in the model.

Suppose that the price of the risk-free asset is the same as in \eqref{riskfreeasset1} and the price of the risky asset is described as
\begin{eqnarray}\label{riskyasset2}
   &&\hspace{-1.7cm}S^{t,S,p;\pmb{\phi}}_s=S+\int_{t}^s \Big(\mu_0+\mu_2 P^{t,p;\pmb{\phi}}_r
   +\pmb{\phi}_{1,r}\sqrt{P^{t,p;\pmb{\phi}}_r}\,\Big)  S^{t,S,p;\pmb{\phi}}_r dr
   +\int_{t}^s \sqrt{P^{t,p;\pmb{\phi}}_r}S_r^{t,S,p;\pmb{\phi}}dB^1_r,
   \\[2mm]\label{volatility1}
   &&\hspace{-1.7cm}P^{t,p;\pmb{\phi}}_s
   =p+\int_t^s\left[\kappa(\overline{P}-P^{t,p;\pmb{\phi}}_r)
   +\left(\rho\pmb{\phi}_{1,r}+\sqrt{1-\rho^2}\pmb{\phi}_{2,r}\right)\sigma\sqrt{P^{t,p;\pmb{\phi}}_r}\,\right]dr
   \nonumber\\[2mm]
   &&\quad+\int_{t}^s \sigma\sqrt{P^{t,p;\pmb{\phi}}_r}\Big(\rho
   dB^1_r+\sqrt{1-\rho^2}dB^2_r\Big),\qquad 0\leq t\leq s\leq T,
\end{eqnarray}
where $T$ is a fixed constant, $B^1$ and $B^2$ are two independent standard Brownian motions  on the probability space $(\Omega,{\cal F},\mathbb{F},\mathbb{P})$. $\sqrt{P^{t,p;\pmb{\phi}}}$ represents the stochastic volatility of the risky asset. $\mu_2$ is a constant and $\mu_2 \sqrt{P^{t,p;\pmb{\phi}}}$ represents the market price of the risk associated with $B^1$. $\kappa,\,\overline{P}$ and $\sigma$ are positive constants satisfying $2\kappa\overline{P}\geq\sigma^2$, which respectively represent the speed of mean reversion in volatility, the long-run level of volatility and the volatility-of-volatility parameter. $\rho\in[-1,1]$ is the correlation coefficient of $B^1$ and the risk of volatility. The ambiguity parameter $\pmb{\phi}=(\pmb{\phi}_1,\pmb{\phi}_2)^{\rm T}$ is the progressively measurable process with respect to $\{{\cal F}_s^t\}_{t\leq s\leq T}$ taking values in $\mathbb{R}^2$.

Let $\pmb{\pi}$ and $1-\pmb{\pi}$ be the proportions of wealth invested respectively in the risky asset and the risk-free asset.
Then the total wealth is expressed as following:
\begin{eqnarray}\label{wealthsde2}
   X_s^{t,x,p;\pmb{\pi},\pmb{\phi}}
   &=&x+\int_t^s\left[\,\mu_0+\Big(\mu_2  P^{t,p;\pmb{\phi}}_r
   +\pmb{\phi}_{1,r}\sqrt{P^{t,p;\pmb{\phi}}_r}\,\Big)\pmb{\pi}_r\,\right]\,
   X_r^{t,x,p;\pmb{\pi},\pmb{\phi}}\,dr
   \nonumber\\[2mm]
   &&+\,\int_t^s\sqrt{P^{t,p;\pmb{\phi}}_r}
   \pmb{\pi}_rX_r^{t,x,p;\pmb{\pi},\pmb{\phi}}dB^1_r, \qquad 0\leq t\leq s\leq T,
\end{eqnarray}
where $x\in(0,+\infty)$ and $\pmb{\pi}$ is \mbox{a} progressively measurable process with respect to $\{{\cal F}_s^t\}_{t\leq s\leq T}$ taking values in $\mathbb{R}$.

The investor chooses the investment strategy $\pmb{\pi}$ to maximize the robust expectation utility $\inf\{{\cal J}_2(t,x;\pmb{\pi},\pmb{\phi}):\pmb{\phi}\in{\Phi}_2(t,p)\}$ with
$$
 {\cal J}_2(t,x;\pmb{\pi},\pmb{\phi})\triangleq
 \mathbb{E}\bigg[\,{\big(X^{t,x,p;\pmb{\pi},\pmb{\phi}}_T\big)^{1-\gamma}\over 1-\gamma}
 +\int_t^T {1-\gamma\over2\theta}|\pmb{\phi}_{s}|^2\,
 W\big(s,X^{t,x,p;\pmb{\pi},\pmb{\phi}}_s,P^{t,p;\pmb{\phi}}_s\big)\,ds\,\bigg],
$$
where the first term in the brace represents the utility of the terminal wealth with the relative risk aversion coefficient $\gamma>0$ and $\gamma\neq1$, the second term represents the entropy penalty and the positive constant $\theta$ measures the strength of the preference for robustness.

Then the stochastic control problem is described as
\begin{equation}\label{valuefunction4}
   W(t,x,p)
   =\sup\limits_{\pmb{\pi}\in{\Pi}_2(t,x,p)}\inf\limits_{\pmb{\phi}\in{\Phi}_2(t,p)}
  {\cal J}_2(t,x;\pmb{\pi},\pmb{\phi}),
\end{equation}
where the admissible sets are taken as follows:
\begin{eqnarray*}
 {\Phi}_2(t,p)=\left\{\,\pmb{\phi}: |\pmb{\phi}|\leq K_\phi \sqrt{P^{t,p;\pmb{\phi}}}\right\},\quad
  {\Pi}_2(t,x,p)=\left\{\pmb{\pi}:|\pmb{\pi}|\leq K_\pi,\,X^{t,x;\pmb{\pi},\pmb{c},\pmb{\phi}}\geq0\right\},
\end{eqnarray*}
with $K_\phi$ and $K_\pi$ are large enough constants.

\begin{theorem}
Suppose that there exist $w\in C^{1,2}([0,T]\times (0,+\infty)^2)$ and $\pi^*,\phi^*$ dependent only on $t,p$ satisfying the following PDEs:
\begin{equation}\label{pde3}
 \left\{
 \begin{array}{l}
   -{\cal L}_2^{\pi^*(t,p),\phi^*(t,p)}w(t,x,p)=f_2(w(t,x,p),\phi^*(t,p));
   \vspace{2mm}\\
   -{\cal L}_2^{\pi,\phi^*(t,p)}w(t,x,p)\geq f_2(w(t,x,p),\phi^*(t,p)),\;\;\forall\;|\pi|\leq K_\pi;
   \vspace{2mm}\\
   -{\cal L}_2^{\pi^*(t,p),\phi}w(t,x,p)\leq f_2(w(t,x,p),\phi),\;\;\forall\;|\phi|\leq K_\phi\sqrt{p};
   \vspace{2mm}\\
   w(T,x,p)={x^{1-\gamma}\over 1-\gamma},
 \end{array}
 \right.
\end{equation}
where the function $f_2$ is defined as
\begin{equation*}
   f_2(w,\phi)\triangleq  {|\phi|^2\over2\theta}(1-\gamma)w,
\end{equation*}
and the differential operator ${\cal L}_2^{\pi,\phi}$ is defined as
\begin{eqnarray*}
   {\cal L}_2^{\pi,\phi} w&\!\!\!\!\triangleq\!\!\!\!&\partial_t w
   +{1\over 2}\pi^2x^2p\partial_{xx}w+\sigma\rho\pi xp\partial_{xp}w
   +{1\over 2}\sigma^2p\partial_{pp}w
   +\Big[\,\mu_0+\Big(\mu_2 p+\phi_1\sqrt{p}\Big)\pi\,\Big]\,x\partial_{x}w
   \\[2mm]
   &&+\Big[\kappa\Big(\overline{P}-p\Big)
   +\Big(\rho\phi_1+\sqrt{1-\rho^2}\phi_2\Big)\sigma\sqrt{p}\,\Big]\,\partial_{p}w.
\end{eqnarray*}

Assume $2\kappa\overline{P}\geq\sigma^2$, and $w,\pi^*,\phi^*$ satisfy the following conditions:
\begin{enumerate}
 \item there exists a constant $K_6$ such that $K_6\geq1$ and $|w|\leq K_6x^{1-\gamma}e^{K_6p}$ on $[0,T]\times(0,+\infty)^2$;

 \item $|\pi^*|\leq K_\pi,\,|\phi^*|\leq K_\phi\sqrt{p}$ on $[0,T]\times[0,+\infty)$. Moreover, $\sqrt{p}\phi^*$ is Lipschitz continuous with respect to $p$.
\end{enumerate}
Then $w=W$ is the unique value function of SCP \eqref{valuefunction4} in space ${\cal B}_{(0,T)}$ with $b_1(t,x,p)=\sum_{i=0}^{N-1}2K_6\overline{w}(t-iT/N,x,p)I_{\{t\in[iT/N,(i+1)T/N)\}}$, where $\overline{w}$ is defined in Theorem A1 with $\varrho=1-\gamma,\;b=2K_6$, and $N$ is a large enough integer such that $N\geq T/\Delta t+(K_1+1)^2+ 300^2|1-\gamma|^2 K_\phi^4Te^{T}/\theta^2$ and $\overline{g}_1(T/N)\geq 1/2,\;\overline{g}_2(T/N)\geq K_6$.
Moreover, $(\pmb{\pi}^*,\pmb{\phi}^*)\triangleq
(\pi^*,\phi^*)(\cdot,P^{t,p;\pmb{\phi}^*})$ is the optimal control strategy.
\end{theorem}

\begin{remark}
If $2\kappa\overline{P}\geq\sigma^2$, $P^{t,p;\pmb{0}}$ is positive a.s. (refer to \cite{Cox}). In the following proof, it will be proved that $P^{t,p;\pmb{\phi}^*}$ and  $X^{t,x,p;\pmb{\pi}^*,\pmb{\phi}^*}$ are also positive a.s. Thus, there is no boundary condition at $x=0$ in PDE \eqref{pde3}.
\end{remark}

\begin{proof}
It is clear that SCP \eqref{valuefunction4} can be described as the form of \eqref{valuefunction1} with
\begin{eqnarray*}
   &\!\!\!\!\!\!&\mathbf{u}=\pmb{\pi},\;\;
   \mathbf{v}=\pmb{\phi}=(\pmb{\phi}_1,\pmb{\phi}_2)^{\rm T},\quad
   {\cal X}=(0,+\infty)^2,\quad
   {\cal U}=[\,-K_\pi,K_\pi\,],\;\;{\cal V}=\mathbb{R},
   \\[2mm]
   &\!\!\!\!\!\!&\alpha_1(t,x,p,\pi,\phi)
   =\Big[\,\mu_0+(\mu_2 p+\phi_1\sqrt{p}\,)\pi\,\Big]\,x,\;\;
   \beta_{11}=\pi x\sqrt{p},\;\;\beta_{12}=0,\;\;
   \beta_{21}=\sigma\rho\sqrt{p},\quad
   \\[2mm]
   &\!\!\!\!\!\!&
   \beta_{22}=\sigma\sqrt{(1-\rho^2)p},\quad
   \alpha_2(t,x,p,\pi,\phi)=\kappa\left(\overline{P}-p\right)
   +\left(\rho\phi_1+\sqrt{1-\rho^2}\phi_2\right)\sigma\sqrt{p},
   \\[2mm]
   &\!\!\!\!\!\!&\Psi(t,x,p)={x^{1-\gamma}\over 1-\gamma},\quad
   f(t,x,p,y,z,w,\pi,\phi)={|\phi|^2\over2\theta}(1-\gamma)w.
\end{eqnarray*}
Set
$$
  K_1=0,\quad b_2(\pi,\phi)={|1-\gamma||\phi|^2\over 2\theta},\quad
   b_3(\pi,\phi)=0,\quad
   b_4(x,p)=K_6x^{1-\gamma}e^{K_6p},
$$
then we can check that Assumption B1-B3 are satisfied.

In order to achieve the existence of the strong solutions of SDE \eqref{volatility1}  and~\eqref{wealthsde2}, we introduce the following Girsanov's measure transformation:
\begin{equation*}
   M^{\pmb{\phi}}_s
   =\exp\left(\,-\int_t^s \pmb{\phi}_{1,r}dB^1_r
   -\int_t^s \pmb{\phi}_{2,r}dB^2_r
   -{1\over2}\,\int_t^s|\pmb{\phi}_r|^2dr\,\right),\quad
   d\mathbb{P}^{\pmb{\phi}}=M^{\pmb{\phi}}_{t+\widehat{\Delta t}}d\mathbb{P}
\end{equation*}
for any $s\in[\,t,t+\widehat{\Delta t}\,]$ and $\pmb{\phi}$ satisfying Novikov's condition (refer to \cite{Karatzas1}), i.e.,
$$
  \mathbb{E}\bigg[\,\exp\Big(\,{1\over2}\,
  \int_t^{t+\widehat{\Delta t}}|\pmb{\phi}_s|^2ds\,\Big)\bigg]<+\infty,
$$
where $\widehat{\Delta t}$ is defined in Theorem A1. We know that $B^{\pmb{\phi},1}$ and $B^{\pmb{\phi},2}$ are Brownian motions on the probability space $(\Omega,{\cal F},\mathbb{F},\mathbb{P}^{\pmb{\phi}})$ with
$$
   B_s^{\pmb{\phi},1}\triangleq\int_t^s\pmb{\phi}_{1,r}\,dr+B^1_s,\quad
   B_s^{\pmb{\phi},2}\triangleq\int_t^s \pmb{\phi}_{2,r}\,dr+B^2_s,\quad
   \forall\;s\in\big[t,t+\widehat{\Delta t}\big].
$$
Considering the following SDEs:
\begin{eqnarray}\label{wealthsde3}
   X^{t,x,p;\pmb{\pi},\pmb{\phi}}_s&=&
   x+\int_{t}^s \left(\mu_0+\mu_2 \pmb{\pi}_rP^{t,p;\pmb{\phi}}_r\,\right)
   X^{t,x,p;\pmb{\pi},\pmb{\phi}}_rdr
   +\int_{t}^s \sqrt{P^{t,p;\pmb{\phi}}_r}\pmb{\pi}_r
   X^{t,x,p;\pmb{\pi},\pmb{\phi}}dB^{\pmb{\phi},1}_r,\qquad
   \\[2mm]\label{volatility2}
   P^{t,p;\pmb{\phi}}_s&=&
   p+\int_t^s\kappa(\overline{P}-P^{t,p;\pmb{\phi}}_r)\,dr
   +\int_{t}^s \sigma\sqrt{P^{t,p;\pmb{\phi}}_r}
   \Big(\rho dB^{\pmb{\phi},1}_r+\sqrt{1-\rho^2}dB^{\pmb{\phi},2}_r\Big).
\end{eqnarray}
For any $(t,x,p)\in[0,T]\times (0,+\infty)^2$, one can easily check that the coefficients of SDE \eqref{volatility2} satisfy locally Lipschitz condition on $(0,+\infty)$. Then, we can deduce that SDE \eqref{volatility2} has a strong solution on a local time interval via a stopping procedure. As $2\kappa\overline{P}\geq\sigma^2$, $P^{t,p;\pmb{\phi}}$ is positive a.s. (refer to \cite{Cox}). Combining with the fact that the coefficients of SDE \eqref{volatility2} satisfy the linear growth condition, we can prove that SDE \eqref{volatility2} has some strong solutions on the whole interval $[t,T]$. The uniqueness of the strong solution to SDE \eqref{volatility2} follows from the fact that the coefficients satisfy the Yamada and Watanabe condition (refer to P291 \cite{Karatzas1}).

Moreover, for any $\pmb{\pi}\in \Pi_2(t,x,p)$, SDE \eqref{wealthsde3} has a unique strong solution (refer to P360 in \cite{Karatzas1}) taking the form of
\begin{equation}\label{wealthsolution}
   X^{t,x,p;\pmb{\pi},\pmb{\phi}}_s
   =x\exp\left\{\,\int_{t}^s \left[\mu_0+\Big(\mu_2 \pmb{\pi}_r-{1\over2}\pmb{\pi}_r^2\Big)\,P^{t,p;\pmb{\phi}}_r\,\right]\,dr
   +\int_{t}^s\pmb{\pi}_r\sqrt{P^{t,p;\pmb{\phi}}_r}dB^{\pmb{\phi},1}_r\,\right\}.
\end{equation}
It is clear that on $[t,t+\widehat{\Delta t}]$, $P^{t,p;\pmb{\phi}}$ and $X^{t,x,p;\pmb{\pi},\pmb{\phi}}$ are the strong solutions of SDE \eqref{volatility1} and \eqref{wealthsde2}, respectively.
Moreover, since the strong solutions of SDE \eqref{volatility2} and~\eqref{wealthsde3} are  also the strong solutions of \eqref{volatility1} and \eqref{wealthsde2}, respectively.
Then by the uniqueness of solutions of SDE \eqref{volatility2} and~\eqref{wealthsde3}, we deduce that SDE \eqref{volatility1} and~\eqref{wealthsde2} have the unique strong solutions on $[t,t+\widehat{\Delta t}]$ for any $\pmb{\phi}$ satisfying Novikov's condition.
In fact, from Theorem A1, we know that for any $\pmb{\phi}\in\Phi_2(t,p)$,
$$
  \mathbb{E}\bigg[\exp\Big(\,{1\over2}\,\int_t^{t+\widehat{\Delta t}}|\pmb{\phi}_s|^2ds\,\Big)\bigg]\leq \mathbb{E}\bigg[\exp\Big({ K_\phi^2\over2}\int_t^{t+\widehat{{\Delta t}}}P^{t,p;\pmb{\phi}}_{s}ds\Big)\bigg]\leq 18e^p<+\infty,
$$
which is exactly the Novikov condition.
Hence,  SDE \eqref{volatility1} and \eqref{wealthsde2} admit unique strong solutions on $[t,t+\widehat{\Delta t}]$ for any $\pmb{\phi}\in\Phi_2(t,p)$ and $\pmb{\pi}\in \Pi_2(t,x,p)$.

Since $\widehat{\Delta t}$ only depends on $\sigma,\kappa,\overline{P},K_\phi$ (see Theorem A1), repeating the same argument, we can establish the existence and uniqueness of the strong solutions to SDE \eqref{volatility1} and \eqref{wealthsde2} on $[t+\widehat{\Delta t},t+2\widehat{\Delta t}]$, with 
the initial values $X_{t+\widehat{\Delta t}}^{t,x,p;\pmb{\pi},\pmb{\phi}}$ and $P_{t+\widehat{\Delta t}}^{t,p;\pmb{\phi}}$, respectively.

By the bootstrap method, we can prove that SDE \eqref{volatility1} and~\eqref{wealthsde2} have unique strong solutions on $[t,T]$. Moreover, it is not difficult to verify that the admissible set $\Pi_2(t,x,p)$ and $\Phi_2(t,p)$ satisfy the ``switching condition''. 

For simplicity, we denote $X^{t,x,p;\pmb{\pi},\pmb{\phi}}$ and $P^{t,p;\pmb{\phi}}$ by $X$ and $P$, respectively. A simple calculation yields
\begin{eqnarray*}
   &&\|b_1(\cdot,X,P)b_2(\cdot,\pmb{\pi},\pmb{\phi})
   +b_3(\cdot,\pmb{\pi},\pmb{\phi})\|_{\mathbb{L}^1(t,(i+1)T/N)}
   +\|b_4(\cdot,X,P)\|_{\mathbb{S}^1(t,(i+1)T/N)}
   \\[2mm]
   &\leq& \mathbb{E}\left[{|1-\gamma|K_6K^2_\phi\over\theta}\sqrt{T\over N}
   \int_t^{{(i+1)T\over N}}{P_{s}\over\sqrt{s-iT/N}}
   \overline{w}\left(s-{iT\over N},X_{s},P_{s}\right)ds\right.
   \\[2mm]
   &&\left.+2K_6\sup\limits_{s\in[t,(i+1)T/N]}
   \overline{w}\left(s-{iT\over N},X_{s},P_{s}\right)\right]
   \leq C\overline{w}(t-iT/N,x,p)\leq Cb_1(t,x,p)
\end{eqnarray*}
for any $t\in[iT/N,(i+1)T/N)$, $i=0,1,\cdots,N-1$, where we have used in the last inequality Theorem  A1 and the fact that SDEs of $X,P$ are time homogenous. Hence, Assumption A2(N) with $q=1$ is satisfied.

Moreover, by Theorem  A1, we know that for any $iT/N\leq t<t+\epsilon\leq(i+1)T/N$, $i=0,1,\cdots,N-1$,
\begin{eqnarray*}
   &&(b_1(t,x,p))^{-1}
   \|b_1(\cdot,X,P)b_2(\cdot,\pmb{\pi},\pmb{\phi})\|_{\mathbb{L}^1(t,t+\epsilon)}
   \\[2mm]
   &\leq& {1\over b_1(t,x,p)}
   \mathbb{E}\left[{|1-\gamma|K_6K^2_\phi \over\theta}\sqrt{T\over N}
   \int_t^{{(i+1)T\over N}}{P_{s}\over\sqrt{s-iT/N}}
   \overline{w}\left(s-{iT\over N},X_{s},P_{s}\right)ds\right]
   \\[2mm]
   &\leq& { 5|1-\gamma|K^2_\phi \over \theta }\sqrt{T\over N}
   \leq {e^{-T/2}\over 60}\leq{e^{-(K_1+1)^2T/(2N)}\over 60},
\end{eqnarray*}
which indicates that Assumption A3(N) with $q=1$ holds.

Then, noting the following relation
$$
  f^w\left(t,x,w(t,x,p),D_{(x,p)}w\cdot\beta(t,x,u,v),u,v\right)=f_2(w(t,x,p),\phi),
$$
and $|w|\leq K_6x^{1-\gamma}e^{K_6p}\leq \min(b_1, b_4)$ on $[0,T]\times(0,+\infty)^2$, we know that it is sufficient to verify Condition 2 in Theorem 2 to get the desired conlusion.
In fact, for any $\pmb{\pi}\in \Pi_2(t,x,p)$, the drift of SDE \eqref{volatility1} with $\pmb{\phi}=\pmb{\phi}^{*;\pmb{\pi}}=\phi^*(\cdot, P^{t,p;\pmb{\phi}^*})$ satisfies Lipschitz condition, and the coefficients of SDE \eqref{volatility1} satisfy Yamada and Watanabe condition. Repeating the above argument, we know that SDE \eqref{volatility1} and \eqref{wealthsde2} have unique strong solutions. As $|\phi^*|\leq K_\phi \sqrt{p}$ on $[0,T]\times[0,+\infty)$, we have $|\pmb{\phi}^{*;\pmb{\pi}}|\leq K_\phi\sqrt{P^{t,p;\pmb{\phi}^{*;\pmb{\pi}}}}$ and $\pmb{\phi}^{*;\pmb{\pi}}\in \Phi_2(t,p)$.

For any $\pmb{\phi}\in \Phi_2(t,p)$, we can repeat the above argument, and introduce the Girsanov's measure transformation, then establish the existence and uniqueness of strong solutions of SDE \eqref{wealthsde3} and \eqref{volatility2}. Finally, we obtain the existence and uniqueness of strong solutions of SDE \eqref{volatility1} and \eqref{wealthsde2} via Girsanov's transformation. Moreover, $|\pi^*|\leq K_\pi$ implies $|\pmb{\pi}^{*;\pmb{\phi}}|\leq K_\pi$. And $X^{t,x;\pmb{\pi},\pmb{c},\pmb{\phi}}\geq0$ comes from \eqref{wealthsolution}. That is, $\pmb{\pi}^{*;\pmb{\phi}}\in\Pi_2(t,x,p)$.

Until now, we have testified all assumptions in Theorem 2 with $q=1$. Hence, by Theorem 2 and Remark \ref{remark of main result2}, the conclusion is obvious.
\end{proof}

Next, we give the explicit form of the value function to finish the example.

\begin{theorem}\label{result2 of example 2}
If $\kappa\geq\max\{\sigma^2/(2\overline{P}),2\sigma|\mu_2|/(\gamma+\theta)\}$ and $K_\pi,\,K_\phi$ are large enough, then the unique value function $W$ in ${\cal B}_{(0,T)}$ and the optimal strategy $(\pi^*,c^*,\phi^*)$ of SCP \eqref{valuefunction4} take the following forms:
\begin{eqnarray*}
   &&W(t,x,p)={x^{1-\gamma}\over
   1-\gamma}\,g_2(t)\,e^{g_3(t)p},\qquad
   \pi^*(t,x,p)={\sigma\rho(1-\gamma-\theta)g_3(t)+\mu_2(1-\gamma)\over(1-\gamma)(\gamma+\theta)},
   \\[2mm]
   &&\phi_1^*(t,x,p)=-{\sigma\rho\,g_3(t)+{\mu_2}(1-\gamma)
   \over(1-\gamma)(\gamma+\theta)}\theta\sqrt{p}\,,\qquad
   \phi_2^*(t,x,p)={-\sigma\theta\,g_3(t)\over 1-\gamma}\,\sqrt{(1-\rho^2)p}\,,
\end{eqnarray*}
where
\begin{equation}\label{timesolution1}
   g_2(t)=\exp\left\{\mu_0(1-\gamma)(T-t)
   +\kappa\overline{P}\int_t^T g_3(r)dr\right\},
\end{equation}
\begin{equation}\label{timesolution2}
   g_3(t)=
   \left\{
   \begin{array}{l}
     {\displaystyle a_2-{a_2(a_2- a_3)\over
     a_2-a_3e^{a_4(a_2-a_3)(t-T)}},\qquad\qquad
     \mbox{if}\;\;\gamma+\theta\neq1};
     \vspace{2mm}\\
     {\displaystyle {a_6\over a_5}\,\Big[\,e^{a_5(T-t)}-1\,\Big],
     \qquad\qquad\qquad\qquad
     \mbox{if}\;\;\gamma+\theta=1,}
     \vspace{2mm}\\
   \end{array}
   \right.
\end{equation}
\begin{equation*}
   a_2={-a_5+\sqrt{a_5^2-4a_4a_6}\over 2a_4},\qquad
   a_3={-a_5-\sqrt{a_5^2-4a_4a_6}\over 2a_4},
\end{equation*}
 $$
   a_4={\sigma^2(1-\gamma-\theta)\over2(1-\gamma)(\gamma+\theta)}
   \Big[\rho^2+(1-\rho^2)(\gamma+\theta)\Big],\;
   a_5={\sigma\rho{\mu_2}\over\gamma+\theta}
   (1-\gamma-\theta)-\kappa, \;
   a_6={(1-\gamma){\mu_2}^2\over 2(\gamma+\theta)}.
$$
\end{theorem}

%

The proof is postponed to Appendix B.

\appendix
\section{Estimates on the SDEs for Heston model}

In this section, we give some estimates on $X^{t,x,p;\pmb{\pi},\pmb{\phi}}$ in \eqref{wealthsde2} and $P^{t,p;\pmb{\phi}}$ in \eqref{volatility1}, which are important for Example 2 in Section 4.2.\smallskip

{\bf Theorem A1} \emph{(1) For any $\varrho\in\mathbb{R},b>0,|\pmb{\pi}|\leq K_\pi, |\pmb{\phi}|\leq K_\phi\sqrt{P^{t,p;\pmb{\phi}}}$, let $\overline{w}(s,x,p)=\overline{g}_1(s)x^{\varrho}\,e^{\overline{g}_2(s)p}$ with $\overline{g}_1(s)=e^{-|\kappa\overline{P}b+\mu_0\varrho|s}$ and $\overline{g}_2(s)=b-4\sqrt{s}$. Then we have the following estimation
\begin{eqnarray*}\nonumber
   \mathbb{E}\bigg[\sup\limits_{s\in[t,\Delta t]}\overline{w}
   \left(s,X^{t,x,p;\pmb{\pi},\pmb{\phi}}_{s},P^{t,p;\pmb{\phi}}_{s}\right)
   +\int_t^{\Delta t}{P^{t,p;\pmb{\phi}}_{s}\over\sqrt{s} }
   \overline{w}\left(s,X^{t,x,p;\pmb{\pi},\pmb{\phi}}_{s},
   P^{t,p;\pmb{\phi}}_{s}\right)ds\bigg]
   \leq 10\overline{w}(t,x,p)
\end{eqnarray*}
for any $(t,x,p)\in[0,\Delta t]\times(0,+\infty)^2$ with
$$
 \Delta t\triangleq
 \min\left\{ \left[\sigma^2b^2+\sigma(2|{\varrho}|K_\pi+\sqrt{2} K_\phi)b
 +({\varrho}^2+|{\varrho}|)K_\pi^2+|{\varrho}{\mu_2}| K_\pi
 +|{\varrho}|K_\phi K_\pi\right]^{-2},\ {b\over64}\right\}.
$$
Moreover, if $0\leq t\leq \widehat{\Delta t}\triangleq
\min\{(\sigma^2+\sqrt{2}\sigma K_\phi)^{-2},1/64,\ln 2/(\kappa\overline{P}),1/K_\phi^2\}$, then we have
$$
   \mathbb{E}\bigg[\exp\Big({ K_\phi^2\over2}
   \int_t^{\widehat{{\Delta t}}}P^{t,p;\pmb{\phi}}_{s}ds\Big)\bigg]
   \leq 18e^p,\;\;\forall\;p\in(0,+\infty).
$$}

\emph{(2) For any $\varrho\in\mathbb{R},T>0,|\pmb{\pi}|\leq K_\pi,|\pmb{\phi}|\leq K_\phi\sqrt{P^{t,p;\pmb{\phi}}}$, if $\kappa/\sigma\geq\overline{a}_1+\sqrt{\overline{a}_2}$ with $\overline{a}_1=|\rho{\varrho}| K_\pi+\left(|\rho|+\sqrt{1-\rho^2}\right)K_\phi$ and $\overline{a}_2=({\varrho}^2-{\varrho})^+K_\pi^2+2|{\varrho}{\mu_2}|K_\pi+2|{\varrho}|K_\phi K_\pi$, then for any constants $k,b$ satisfying $0\leq 2k\leq(\kappa/\sigma-\overline{a}_1)^2-\overline{a}_2$ and
$$
 \left({\kappa\over\sigma}-\overline{a}_1\right)-\sqrt{\Delta_1}
 \leq \sigma b
 \leq \left({\kappa\over\sigma}-\overline{a}_1\right)+\sqrt{\Delta_1},\quad
 \Delta_1=\left({\kappa\over\sigma}-\overline{a}_1\right)^2-(\overline{a}_2+2k),
$$
we have the following estimation
$$
  \sup\limits_{s\in[t,T]}
  \mathbb{E}\left[\left(X^{t,x,p;\pmb{\pi},\pmb{\phi}}_{s}\right)^\varrho
  e^{bP^{t,p;\pmb{\phi}}_{s}}\right]
  +k\mathbb{E}\left[\int_t^TP^{t,p;\pmb{\phi}}_{s}
  \left(X^{t,x,p;\pmb{\pi},\pmb{\phi}}_{s}\right)^\varrho
  e^{bP^{t,p;\pmb{\phi}}_{s}}ds\right]
  \leq 5e^{|\kappa\overline{P}b+\mu_0\varrho|T}x^{\varrho}\,e^{bp}
$$
for any $(t,x,p)\in[0,T]\times(0,+\infty)^2$. Particularly, if $k>0$, then
$$
  \mathbb{E}\bigg[\sup\limits_{s\in[t,T]}
  \Big(X^{t,x,p;\pmb{\pi},\pmb{\phi}}_{s}\Big)^\varrho
  e^{bP^{t,p;\pmb{\phi}}_{s}}\bigg]
  \leq 3\left[k+(|\varrho|K_\pi+\sigma|\rho|b)^2+\sigma^2(1-\rho^2)b^2\right]
  {e^{|\kappa\overline{P}b+\mu_0\varrho|T}\over k}\,x^{\varrho}e^{bp}.
$$}

\emph{(3) In particular, in the non-ambiguity model, for any $\varrho\in\mathbb{R},T>0$, if $\kappa/\sigma>|\rho{\varrho}|+\sqrt{({\varrho}^2-{\varrho})^++2|{\varrho}{\mu_2}|}$, then for any
$$
 b\in\bigg({\kappa/\sigma-|\rho \varrho|-\sqrt{\Delta_2}\over\sigma},
 {\kappa/\sigma-|\rho \varrho|+\sqrt{\Delta_2}\over\sigma}\bigg),$$
 with
$\Delta_2=\left({\kappa\over\sigma}-|\rho \varrho|\right)^2-\big[({\varrho}^2-{\varrho})^+
 +2|{\varrho}{\mu_2}|\big],$
there exists a constant $C$ such that
$$
  \mathbb{E}\bigg[\sup\limits_{s\in[t,T]}
  \left(S^{t,S,p}_{s}\right)^\varrho e^{bP^{t,p}_{s}}
   +\int_t^TP^{t,p}_{s}\left(S^{t,S,p}_{s}\right)^\varrho e^{bP^{t,p}_{s}}ds\bigg]
  \leq CS^{\varrho}\,e^{bp}
$$
for any $(t,S,p)\in[0,T]\times(0,+\infty)^2$, where $S^{t,S,p}$ and $P^{t,p}$ are the price of risky asset satisfying SDE \eqref{riskyasset2} and the volatility satisfying SDE \eqref{volatility1} with $\pmb{\phi}=(0,0)$, respectively.}

\begin{proof}
(1) It is not difficult to check that the functions $\overline{g}_1$ and $\overline{g}_2$ satisfy
$$
 \overline{g}^\prime_1,\overline{g}^\prime_2\leq0,\quad
 0<\overline{g}_1\leq 1,\quad{b\over2}\leq\overline{g}_2\leq b,\quad
 {\overline{g}_1^\prime\over \overline{g}_1}+{\mu_0}{\varrho}+\kappa\overline{P}\overline{g}_2\leq 0\;\;
 \mbox{on}\;[0,\Delta t].
$$
A simple computation gives
\begin{eqnarray*}
   &&\partial_s \overline{w}=\overline{g}^\prime_1(s)x^\varrho\,
   e^{\overline{g}_2(s)p}+\overline{g}_1(s)x^\varrho\,
   e^{\overline{g}_2(s)p}\overline{g}^\prime_2(s)p
   =\left[\,{\overline{g}_1^\prime(s)\over \overline{g}_1(s)}+\overline{g}^\prime_2(s)p\,\right] \overline{w},
   \\[2mm]
   &&\partial_x \overline{w}={\varrho} \overline{g}_1(s)x^{{\varrho}-1}\,e^{\overline{g}_2(s)p}={{\varrho}\overline{w}\over x},
   \quad x\partial_x \overline{w}={\varrho} \overline{w},
   \\[2mm]
   &&\partial_{xx} \overline{w}={\varrho}({\varrho}-1) \overline{g}_1(s)x^{{\varrho}-2}\,e^{\overline{g}_2(s)p}
   ={\varrho}({\varrho}-1){\overline{w}\over x^2},\quad x^2\partial_{xx} \overline{w}={\varrho}({\varrho}-1)\overline{w},
   \\[2mm]
   &&\partial_p \overline{w}=\overline{g}_1(s)x^{\varrho}\,
   e^{\overline{g}_2(s)p}\overline{g}_2(s)=\overline{g}_2(s)\overline{w},\quad
   x\partial_{px}\overline{w}=x\overline{g}_2(s)\partial_x\overline{w}={\varrho} \overline{g}_2(s)\overline{w},
   \\[2mm]
   &&\partial_{pp} \overline{w}=\partial_p(\overline{g}_2(s)\overline{w})
   =\overline{g}_2(s)\partial_p\overline{w}=\overline{g}^2_2(s)\overline{w}.
\end{eqnarray*}
Since
\begin{eqnarray*}
   {\cal L}_2^{\pi,\phi} \overline{w}&=&
   \partial_s \overline{w}+{1\over2}\pi^2x^2p\partial_{xx}\overline{w}
   +\sigma\rho\pi xp\partial_{xp}\overline{w}
   +{1\over2}\sigma^2p\partial_{pp}\overline{w}
   +\Big[\,{\mu_0}+\Big({\mu_2} p+\phi_1\sqrt{p}\Big)\pi\,\Big]\,x\partial_{x}\overline{w}
   \\[2mm]
   &&+\,\Big[\kappa\Big(\overline{P}-p\Big)+\Big(\rho\phi_1
   +\sqrt{1-\rho^2}\phi_2\Big)\sigma\sqrt{p}\,\Big]\,\partial_{p}\overline{w},
\end{eqnarray*}
we have
\begin{eqnarray*}
   {{\cal L}_2^{\pi,\phi}\overline{w}\over \overline{w}}+{p\over\sqrt{s} }
   &=&\left[{\overline{g}_1^\prime(s)\over \overline{g}_1(s)}
   +{\mu_0}{\varrho}+\kappa\overline{P}\overline{g}_2(s)\right]
   +p\left[\overline{g}_2^\prime(s)+{1\over 2}\,{\varrho}({\varrho}-1)\pi^2
   +\sigma\rho{\varrho} \overline{g}_2(s)\pi+{\varrho}{\mu_2}\pi\right.
   \\[2mm]
   &&\left.+{1\over2}\sigma^2\overline{g}^2_2(s)
   +{\varrho}\pi{\phi_1\over \sqrt{p}}-\kappa \overline{g}_2(s)
   +\sigma\rho\overline{g}_2(s){\phi_1\over \sqrt{p}}
   +\sigma\sqrt{1-\rho^2}\overline{g}_2(s){\phi_2\over \sqrt{p}}
   +{1\over\sqrt{s}}\right]
   \\[2mm]
   &\leq& p\left[-{2\over\sqrt{s}}+{1\over 2}({\varrho}^2-{\varrho})^+K_\pi^2
   +\sigma|{\varrho}|K_\pi b+|{\varrho}{\mu_2}| K_\pi+{1\over2}\sigma^2b^2
   +|{\varrho}|K_\phi K_\pi\right.
   \\[2mm]
   &&\left.+\sqrt{2}\sigma K_\phi b+{1\over\sqrt{s}}\right]\leq0,\;\quad
   \mbox{on}\;\;[0,\Delta t]\times(0,+\infty)^2,
\end{eqnarray*}
where we have used the definition of $\Delta t$ in the last inequality. By It\^o's formula, we have, for any $t\leq s\leq\Delta t$,
\begin{eqnarray}\nonumber
 0&\!\!\!\leq\!\!\!&
 \overline{w}\left(s,X^{t,x,p;\pmb{\pi},\pmb{\phi}}_{s},P^{t,p;\pmb{\phi}}_{s}\right)
 +\int_t^s {P^{t,p;\pmb{\phi}}_{r}\over\sqrt{r}} \overline{w}\left(r,X^{t,x,p;\pmb{\pi},\pmb{\phi}}_{r},P^{t,p;\pmb{\phi}}_{r}\right)dr
 \\[2mm]\nonumber
 &\!\!\!=\!\!\!&\overline{w}(t,x,p)
 +\int_t^s\left ({\cal L}_2^{\pi,\phi}\overline{w}+{p\over\sqrt{r}} \overline{w}\right)\left(r,X^{t,x,p;\pmb{\pi},\pmb{\phi}}_{r},P^{t,p;\pmb{\phi}}_{r};
 \pmb{\pi}_{r},\pmb{\phi}_{r}\right)dr
 +\int_t^s {\cal G}^1_{r}dB^1_{r}+\int_t^s {\cal G}^2_{r}dB^2_{r}
 \\[2mm]\label{equality in Appendix1}
 &\!\!\!\leq\!\!\!&\overline{w}(t,x,p)+\int_t^s {\cal G}^1_{r}dB^1_{r}
 +\int_t^s {\cal G}^2_{r}dB^2_{r},
\end{eqnarray}
with
\begin{eqnarray*}
 {\cal G}^1_{r}&=&\left[\varrho \pmb{\pi}_{r} +\sigma\rho\overline{g}_2(r)\right]
 \sqrt{P^{t,p;\pmb{\phi}}_{r}}
 \overline{w}\left(r,X^{t,x,p;\pmb{\pi},\pmb{\phi}}_{r},P^{t,p;\pmb{\phi}}_{r}\right),
 \\[2mm]
 {\cal G}^1_{r}&=&\sigma\sqrt{1-\rho^2}\overline{g}_2(r)\sqrt{P^{t,p;\pmb{\phi}}_{r}}
 \overline{w}\left(r,X^{t,x,p;\pmb{\pi},\pmb{\phi}}_{r},P^{t,p;\pmb{\phi}}_{r}\right).
\end{eqnarray*}
Since the right hand of \eqref{equality in Appendix1} is a non-negative local martingale, then Fatou's lemma implies that
\begin{equation}\label{dd}
 \mathbb{E}\bigg[
 \overline{w}\left(s,X^{t,x,p;\pmb{\pi},\pmb{\phi}}_{s},P^{t,p;\pmb{\phi}}_{s}\right)
 +\int_t^s {P^{t,p;\pmb{\phi}}_{r}\over\sqrt{r}} \overline{w}\left(r,X^{t,x,p;\pmb{\pi},\pmb{\phi}}_{r},P^{t,p;\pmb{\phi}}_{r}\right)dr
 \bigg]
 \leq \overline{w}(t,x,p)
\end{equation}
for any $0\leq t\leq s\leq \Delta t$ and any $(x,p)\in(0,+\infty)^2$.

Moreover, it is not difficult to compute that
\begin{eqnarray*}
   {2{\cal L}_2^{\pi,\phi}\sqrt{\overline{w}}\over \sqrt{\overline{w}}}
   &=&\left[{\overline{g}_1^\prime(s)\over \overline{g}_1(s)}
   +{\mu_0}{\varrho}+\kappa\overline{P}\overline{g}_2(s)\right]
   +p\left[\overline{g}_2^\prime(s)
   +{1\over 2}\,{\varrho}\left({\varrho\over2}-1\right)\pi^2
   +{1\over 2}\,\sigma\rho{\varrho} \overline{g}_2(s)\pi+{\varrho}{\mu_2}\pi\right.
      \\[2mm]
   &&\left.+{1\over4}\sigma^2\overline{g}^2_2(s)
   +{\varrho}\pi{\phi_1\over \sqrt{p}}
   -\kappa \overline{g}_2(s)+\sigma\rho\overline{g}_2(s){\phi_1\over \sqrt{p}}
   +\sigma\sqrt{1-\rho^2}\overline{g}_2(s){\phi_2\over \sqrt{p}}\right]\leq0
 \end{eqnarray*}
on $[0,\Delta t]\times(0,+\infty)^2$. Repeating the above argument, we can deduce that
\begin{eqnarray*}
 0\leq
 \sqrt{\overline{w}}\left(s,X^{t,x,p;\pmb{\pi},\pmb{\phi}}_{s},P^{t,p;\pmb{\phi}}_{s}\right)
 \leq\sqrt{\overline{w}}(t,x,p)+\int_t^s {\cal G}^3_{r}dB^1_{r}
 +\int_t^s {\cal G}^4_{r}dB^2_{r},\;\;\  0\leq t\leq s\leq \Delta t,
\end{eqnarray*}
with
\begin{eqnarray*}
   {\cal G}^3_{r}&=&{1\over2}\left[\varrho \pmb{\pi}_{r}+\sigma\rho\overline{g}_2(r)\right]
   \sqrt{P^{t,p;\pmb{\phi}}_{r}}
   \sqrt{\overline{w}}\left(r,X^{t,x,p;\pmb{\pi},\pmb{\phi}}_{r},P^{t,p;\pmb{\phi}}_{r}\right),
   \\[2mm]
   {\cal G}^4_{r}&=&{1\over2}\sigma\sqrt{1-\rho^2}\overline{g}_2(r)
   \sqrt{P^{t,p;\pmb{\phi}}_{r}}
   \sqrt{\overline{w}}\left(r,X^{t,x,p;\pmb{\pi},\pmb{\phi}}_{r},P^{t,p;\pmb{\phi}}_{r}\right).
\end{eqnarray*}

From the Burkholder-Davis-Gundy inequality, we know that for any $0\leq t\leq \Delta t$,
\begin{eqnarray*}
 &&\mathbb{E}\bigg[\sup\limits_{s\in[t,\Delta t]}
 \Big|\int_t^s {\cal G}^3_{r}dB^1_{r}\Big|^2\bigg]
 \leq4\mathbb{E}\bigg[\int_t^{\Delta t} \left|{\cal G}^3_{r}\right|^2dr\bigg]
 \\[2mm]
 &&\hspace{1.5cm}\leq(|\varrho|K_\pi+\sigma|\rho|b)^2\sqrt{\Delta t}
 \mathbb{E}\bigg[\int_t^{\Delta t}{P^{t,p;\pmb{\phi}}_{s}\over\sqrt{s} }
 \overline{w}\left(s,X^{t,x,p;\pmb{\pi},\pmb{\phi}}_{s},
 P^{t,p;\pmb{\phi}}_{s}\right)ds\bigg]
 \leq \overline{w}(t,x,p),
 \\[2mm]
 &&\mathbb{E}\bigg[\sup\limits_{s\in[t,\Delta t]}
 \Big|\int_t^s {\cal G}^4_{r}dB^2_{r}\Big|^2\bigg]
 \leq4\mathbb{E}\bigg[\int_t^{\Delta t} \left|{\cal G}^4_{r}\right|^2dr\bigg]
 \leq\overline{w}(t,x,p),
\end{eqnarray*}
where we have used \eqref{dd} and the definition of $\Delta t$. Then, we deduce that for any $(t,x,p)\in[0,\Delta t]\times(0,+\infty)^2$,
\begin{eqnarray*}
  0&\leq& \mathbb{E}\bigg[\sup\limits_{s\in[t,\Delta t]} \overline{w}\Big(s,X^{t,x,p;\pmb{\pi},\pmb{\phi}}_{s},P^{t,p;\pmb{\phi}}_{s}\Big)\bigg]
  =\mathbb{E}\bigg\{\sup\limits_{s\in[t,\Delta t]} \left[\sqrt{\overline{w}}\left(s,X^{t,x,p;\pmb{\pi},\pmb{\phi}}_{s},
  P^{t,p;\pmb{\phi}}_{s}\right)\right]^2\bigg\}
  \\[2mm]
  &\leq&3\mathbb{E}\bigg\{\sup\limits_{s\in[t,\Delta t]} \Big[ \overline{w}(t,x,p)
  +\Big|\int_t^s {\cal G}^3_{r}dB^1_{r}\Big|^2
  +\Big|\int_t^s {\cal G}^4_{r}dB^2_{r}\Big|^2\Big]\bigg\}
  \leq9\overline{w}(t,x,p).
\end{eqnarray*}

Taking $\varrho=0,b=1$ in $\overline{w}(t,x,p)$, we can check that $\widehat{{\Delta t}}\leq \Delta t,\ 2\overline{w}(t,x,p)\geq e^{K_\phi^2\widehat{\Delta t} p/2}$  and
\begin{eqnarray*}
   &&\mathbb{E}\bigg[\exp\Big({ K_\phi^2\over2}
   \int_t^{\widehat{{\Delta t}}}P^{t,p;\pmb{\phi}}_{s}ds\Big)\bigg]
   \leq\mathbb{E}\bigg[\exp\Big({K_\phi^2\widehat{\Delta t}\over2}
   \sup\limits_{s\in[t,\widehat{\Delta t}]}P^{t,p;\pmb{\phi}}_{s}\Big)\bigg]
   \\[2mm]
   &&\hspace{-4mm}=\mathbb{E}\bigg[\sup\limits_{s\in[t,\widehat{\Delta t}]}
   \exp\Big({K_\phi^2\widehat{\Delta t}\over2}P^{t,p;\pmb{\phi}}_{s}\Big)\bigg]
   \leq2\mathbb{E}\bigg[\sup\limits_{s\in[t,\widehat{\Delta t}]}
   \overline{w}\Big(s,X^{t,x,p;\pmb{\pi},\pmb{\phi}}_{s},P^{t,p;\pmb{\phi}}_{s}\Big)\bigg]
   \\[2mm]
   &&\hspace{-4mm}\leq 18\overline{w}(t,x,p)=18e^p,\quad
   \mbox{on}\;\;[0,\widehat{\Delta t}]\times(0,+\infty)^2.
\end{eqnarray*}

(2) Let $\widetilde{w}(s,x,p)=\overline{g}_3(s)x^\varrho e^{bp}$ with $ \overline{g}_3(s)=e^{|\kappa\overline{P}b+\mu_0\varrho|(T-s)}$ on $[0,T]\times(0,+\infty)^2$. Repeating the above argument, we have
\begin{eqnarray*}
   {{\cal L}_2^{\pi,\phi}\widetilde{w}\over \widetilde{w}}+kp
   &=&\left[{\overline{g}_3^\prime(s)\over \overline{g}_3(s)}
   +{\mu_0}{\varrho}+\kappa\overline{P}b\right]
   +p\left[{1\over 2}\,{\varrho}({\varrho}-1)\pi^2+\sigma\rho{\varrho}b\pi
   +{1\over2}\sigma^2b^2+{\varrho}{\mu_2} \pi+{\varrho}\pi{\phi_1\over\sqrt{p}}\right.
   \\[2mm]
   &&\left.-\kappa b+\sigma\rho\,b\,{\phi_1\over\sqrt{p}}
   +\sigma\sqrt{1-\rho^2}\,b\,{\phi_2\over\sqrt{p}}+k\right]
   \\[2mm]
   &\leq&p\left[{1\over2}\sigma^2b^2
   -\left({\kappa\over\sigma}-\overline{a}_1\right)\sigma b
   +{1\over 2}(\overline{a}_2+2k)\right]\leq0
 \end{eqnarray*}
and
\begin{eqnarray*}
  &&\mathbb{E}\left\{\left(X^{t,x,p;\pmb{\pi},\pmb{\phi}}_{s}\right)^\varrho
  e^{bP^{t,p;\pmb{\phi}}_{s}}
  + k\int_t^sP^{t,p;\pmb{\phi}}_{r}\left(X^{t,x,p;\pmb{\pi},\pmb{\phi}}_{r}\right)^\varrho
  e^{bP^{t,p;\pmb{\phi}}_{r}}dr\right\}
  \\[2mm]
  &\leq&\mathbb{E}\left[
  \widetilde{w}\left(s,X^{t,x,p;\pmb{\pi},\pmb{\phi}}_{s},P^{t,p;\pmb{\phi}}_{s}\right)
  +k\int_t^sP^{t,p;\pmb{\phi}}_{r} \widetilde{w}\left(r,X^{t,x,p;\pmb{\pi},\pmb{\phi}}_{r},P^{t,p;\pmb{\phi}}_{r}\right)dr
  \right]
  \\[2mm]
  &\leq&\widetilde{w}(t,x,p)
  \leq e^{|\kappa\overline{P}b+\mu_0\varrho|T}x^{\varrho}\,e^{bp}.
\end{eqnarray*}

Moreover, we can check that  ${\cal L}_2^{\pi,\phi}\sqrt{\widetilde{w}}\leq0$, and
\begin{eqnarray*}
 0\leq\sqrt{\widetilde{w}}
 \left(s,X^{t,x,p;\pmb{\pi},\pmb{\phi}}_{s},P^{t,p;\pmb{\phi}}_{s}\right)
 \leq\sqrt{\widetilde{w}}(t,x,p)+\int_t^s {\cal G}^5_{r}dB^1_{r}
 +\int_t^s {\cal G}^6_{r}dB^2_{r},\quad t\leq s\leq T,
\end{eqnarray*}
with
\begin{eqnarray*}
  {\cal G}^5_{r}&=&{1\over2}(\varrho\pmb{\pi}_{r} +\sigma\rho b)
  \sqrt{P^{t,p;\pmb{\phi}}_{r}}
  \sqrt{\widetilde{w}}\left(r,X^{t,x,p;\pmb{\pi},\pmb{\phi}}_{r},P^{t,p;\pmb{\phi}}_{r}\right),
  \\[2mm]
  {\cal G}^6_{r}&=&{1\over2}\sigma b\sqrt{1-\rho^2}\sqrt{P^{t,p;\pmb{\phi}}_{r}}
  \sqrt{\widetilde{w}}\left(r,X^{t,x,p;\pmb{\pi},\pmb{\phi}}_{r},P^{t,p;\pmb{\phi}}_{r}\right).
\end{eqnarray*}

By the Burkholder-Davis-Gundy inequality, we obtain
\begin{eqnarray*}
 &&\mathbb{E}\bigg[\sup\limits_{s\in[t,T]}
 \Big|\int_t^s {\cal G}^5_{r}dB^1_{r}\Big|^2\bigg]
 \leq4\mathbb{E}\left[\int_t^T\left|{\cal G}^5_{r}\right|^2dr\right]
 \\[2mm]
 &\leq&(|\varrho|K_\pi+\sigma|\rho|b)^2
 \mathbb{E}\left[\int_t^TP^{t,p;\pmb{\phi}}_{s}\widetilde{w}
 \left(s,X^{t,x,p;\pmb{\pi},\pmb{\phi}}_{s},P^{t,p;\pmb{\phi}}_{s}\right)ds\right]
 \\[2mm]
 &\leq&{(|\varrho|K_\pi+\sigma|\rho|b)^2\over k}\,
 e^{|\kappa\overline{P}b+\mu_0\varrho|T}x^{\varrho}\,e^{bp}\end{eqnarray*}
and
\begin{eqnarray*}\mathbb{E}\bigg[\sup\limits_{s\in[t,T]}
 \Big|\int_t^s {\cal G}^6_{r}dB^1_{r}\Big|^2\bigg]
 \leq{\sigma^2(1-\rho^2)b^2\over k}\,
 e^{|\kappa\overline{P}b+\mu_0\varrho|T}x^{\varrho}\,e^{bp}\end{eqnarray*}
 and
\begin{eqnarray*}
 0&\leq& \mathbb{E}\bigg\{\sup\limits_{s\in[t,T]}
 \Big(X^{t,x,p;\pmb{\pi},\pmb{\phi}}_{s}\Big)^\varrho
 e^{bP^{t,p;\pmb{\phi}}_{s}}\bigg\}
 \leq\mathbb{E}\bigg[\sup\limits_{s\in[t,T]}
 \widetilde{w}\left(s,X^{t,x,p;\pmb{\pi},\pmb{\phi}}_{s},P^{t,p;\pmb{\phi}}_{s}\right)\bigg]
 \\[2mm]
 &\leq&3\bigg[1+{(|\varrho|K_\pi+\sigma|\rho|b)^2\over k}
 +{\sigma^2(1-\rho^2)b^2\over k}\bigg]
 \,e^{|\kappa\overline{P}b+\mu_0\varrho|T}x^{\varrho}\,e^{bp}.
\end{eqnarray*}

(3) In the non-ambiguity model, $K_\phi=0$. Moreover, since $S^{t,S,p}=X^{t,S,p;\pmb{1},\pmb{0}}$, we only have to estimate $S^{t,S,p}$ and $P^{t,p}$, then $K_\pi$ can be chosen equal to 1. Hence,  the conclusion can be deduced from Conclusion (2).
\end{proof}

\section{Proof of Theorem 6}

In this section, we give the proof of Theorem 6.\smallskip

\noindent{\bf Proof of Theorem 6. } Let
$$
 \overline{w}(t,x,p)={1\over 1-\gamma}\,g_2(t)x^{1-\gamma}\,e^{g_3(t)p}\,,
$$
where $g_2(t)$ and $g_3(t)$ will be defined later.

A simple computation gives
\begin{eqnarray*}
   &&\partial_t w
   ={1\over 1-\gamma}\,g^\prime_2(t) x^{1-\gamma}\,e^{g_3(t)p}
   +{1\over 1-\gamma}\,g_2(t)x^{1-\gamma}\,e^{g_3(t)p}g^\prime_3(t)p
   =\left[\,{g_2^\prime(t)\over g_2(t)}+g^\prime_3(t)p\,\right] w,
   \\[2mm]
   &&\partial_x w=g_2(t)x^{-\gamma}\,e^{g_3(t)p}=(1-\gamma){w\over x},
   \quad x\partial_x w=(1-\gamma)w,
   \\[2mm]
   &&\partial_{xx} w=-\gamma g_2(t)x^{-1-\gamma}\,e^{g_3(t)p}
   =-\gamma(1-\gamma){w\over x^2},\quad
   x^2\partial_{xx} w=-\gamma(1-\gamma)w,
   \\[2mm]
   &&\partial_p w={1\over1-\gamma}\,g_2(t)x^{1-\gamma}\,e^{g_3(t)p}g_3(t)=g_3(t)w,\quad
   x\partial_{px}w=x g_3(t)\partial_xw=(1-\gamma)g_3(t)w,
   \\[2mm]
   &&\partial_{pp} w=\partial_p(g_3(t)w)=g_3(t)\partial_pw=g^2_3(t)w,
\end{eqnarray*}
and
\begin{equation*}
   {{\cal L}_2^{\pi,\phi}+f_2\over w}
   =\left[{g_2^\prime(t)\over g_2(t)}
   +{\mu_0}(1-\gamma)+\kappa\overline{P}g_3(t)\right]
   +p\Big[g_3^\prime(t)
   +{1\over2}\sigma^2g_3^2(t)-\kappa g_3(t)
   +F(\pi,\widehat{\phi}_1,\widehat{\phi}_2)\Big],
\end{equation*}
with
$\widehat{\phi}_1=\phi_1/\sqrt{p},\,\widehat{\phi}_2=\phi_2/\sqrt{p}$,
and
\begin{eqnarray*}
   F(\pi,\widehat{\phi}_1,\widehat{\phi}_2)
   &=&{1\over 2}\,\gamma(\gamma-1)\pi^2
   +\sigma\rho(1-\gamma)g_3(t)\pi+(1-\gamma){\mu_2} \pi+(1-\gamma)\widehat{\phi}_1\pi
   \\[2mm]
   &&+\,\sigma\rho\widehat{\phi}_1\,g_3(t)+\sigma\sqrt{1-\rho^2}\,\widehat{\phi}_2\,g_3(t)
   +{1-\gamma\over2\theta}\widehat{\phi}_1^2+{1-\gamma\over2\theta}\widehat{\phi}_2^2\,.
\end{eqnarray*}

It is clear that the following equalities hold at the maximiser $(\pi^*,\widehat{\phi}^*_1,\widehat{\phi}^*_2)$:
\begin{eqnarray*}
   &&(a)\ \ F_\pi=-\gamma(1-\gamma)\pi+\sigma\rho(1-\gamma)g_3(t)+(1-\gamma){\mu_2} +(1-\gamma)\widehat{\phi}_1=0,
   \\[2mm]
   &&(b)\ \ F_{\widehat{\phi}_1}=
   (1-\gamma)\pi+\sigma\rho g_3(t)+{1-\gamma\over\theta}\widehat{\phi}_1=0,
   \\[2mm]
   &&(c)\ \ F_{\widehat{\phi}_2}=
   \sigma\sqrt{1-\rho^2}\,g_3(t)+{1-\gamma\over\theta}\widehat{\phi}_2=0.
\end{eqnarray*}
Solving (c), we get
\begin{equation*}
   \widehat{\phi}_2^*
   ={-\sigma\theta\over 1-\gamma}\,\sqrt{1-\rho^2}\,g_3(t)
   ={-\sigma\theta\,g_3(t)\over 1-\gamma}\,\sqrt{1-\rho^2}.
\end{equation*}
Hence,
\begin{equation*}
    \phi_2^*
   ={-\sigma\theta\,g_3(t)\over 1-\gamma}\,\sqrt{(1-\rho^2)p}\,.
\end{equation*}
From (b), we know
\begin{equation*}
   \widehat{\phi}_1^*
   =-\theta\pi^*-{\sigma\theta\rho\over 1-\gamma}\,g_3(t).
\end{equation*}
Combining with (a), we obtain the following relation
\begin{equation*}
   -\gamma(1-\gamma)\pi^*-(1-\gamma)\theta\pi^*
   =-\sigma\rho(1-\gamma)g_3(t)-(1-\gamma){\mu_2}
   +\sigma\theta\rho g_3(t),
\end{equation*}
which yields
\begin{equation*}
  \pi^*={\sigma\rho(1-\gamma-\theta)\over(1-\gamma)(\gamma+\theta)}\,g_3(t)
   +{{\mu_2} \over \gamma+\theta}
   ={\sigma\rho(1-\gamma-\theta)g_3(t)+\mu_2(1-\gamma)\over(1-\gamma)(\gamma+\theta)}.
\end{equation*}
Therefore,
\begin{equation*}
   \widehat{\phi}_1^*=
   {-\sigma\theta\rho\over(1-\gamma)(\gamma+\theta)}\,g_3(t)
   -{\theta{\mu_2} \over \gamma+\theta}
   =-{\sigma\rho\,g_3(t)+{\mu_2}(1-\gamma)\over(1-\gamma)(\gamma+\theta)}\theta,
\end{equation*}
it follows that
\begin{equation*}
   \phi_1^*=
   -{\sigma\rho\,g_3(t)+{\mu_2}(1-\gamma)\over(1-\gamma)(\gamma+\theta)}\theta\sqrt{p}\,.
\end{equation*}

Hence, $F(\pi^*,\widehat{\phi}^*_1,\widehat{\phi}^*_2)$ can be written as
\begin{eqnarray*}
   F(\pi^*,\widehat{\phi}^*_1,\widehat{\phi}^*_2)
   =\mathfrak{F}_1 g_3^2(t)+\mathfrak{F}_2 g_3(t)+\mathfrak{F}_3
\end{eqnarray*}
with
\begin{eqnarray*}
   \mathfrak{F}_1&=&
   -{\gamma\sigma^2\rho^2(1-\gamma-\theta)^2\over 2(1-\gamma)(\gamma+\theta)^2}
   +{\sigma^2\rho^2(1-\gamma-\theta)\over \gamma+\theta}
   -{\sigma^2\theta\rho^2(1-\gamma-\theta)\over (1-\gamma)(\gamma+\theta)^2}
   -{\sigma^2\theta\rho^2\over (1-\gamma)(\gamma+\theta)}
   \\[2mm]
   &&-\,{\sigma^2\theta(1-\rho^2)\over 1-\gamma}
   +{\sigma^2\theta\rho^2\over 2(1-\gamma)(\gamma+\theta)^2}
   +{\sigma^2\theta(1-\rho^2)\over 2(1-\gamma)},
   \\[2mm]
   \mathfrak{F}_2&=&-{\gamma\sigma\rho{\mu_2} (1-\gamma-\theta)\over (\gamma+\theta)^2}
   +{\sigma\rho{\mu_2} (1-\gamma)\over \gamma+\theta}
   +{\sigma\rho{\mu_2} (1-\gamma-\theta)\over \gamma+\theta}
   -{\sigma\theta\rho{\mu_2} (2-\gamma-\theta)\over
   (\gamma+\theta)^2}
   \\[2mm]
   &&-\,{\sigma\theta\rho{\mu_2} \over \gamma+\theta}
   +\,{\sigma\theta\rho{\mu_2} \over (\gamma+\theta)^2},
   \\[2mm]
   \mathfrak{F}_3&=&-{\gamma(1-\gamma){\mu_2}^2\over 2(\gamma+\theta)^2}
   +{(1-\gamma){\mu_2}^2\over \gamma+\theta}
   -{(1-\gamma)\theta{\mu_2}^2\over (\gamma+\theta)^2}
   +{(1-\gamma)\theta{\mu_2}^2\over 2(\gamma+\theta)^2}.
\end{eqnarray*}
By a careful calculation, it is not difficult to derive
\begin{eqnarray*}
   \mathfrak{F}_1&=&
   (\mathfrak{F}_{1,1}+\mathfrak{F}_{1,3}+\mathfrak{F}_{1,6})
   +\mathfrak{F}_{1,2}+\mathfrak{F}_{1,4}+(\mathfrak{F}_{1,5}
   +\mathfrak{F}_{1,7})
   \\[2mm]
   &=&{\sigma^2\rho^2(-\gamma^2-\gamma\theta+2\gamma+2\theta-1)\over 2(1-\gamma)(\gamma+\theta)}
   +\mathfrak{F}_{1,2}+\mathfrak{F}_{1,4}
   +\Big[\,{\sigma^2\theta\rho^2\over 2(1-\gamma)}-{\sigma^2\theta\over
   2(1-\gamma)}\,\Big]
   \\[2mm]
   &=&{\sigma^2\rho^2(1-\gamma-\theta)^2\over 2(1-\gamma)(\gamma+\theta)}
   -{\sigma^2\theta\over 2(1-\gamma)}=a_4-{\sigma^2\over 2},
   \\[2mm]
   \mathfrak{F}_2&=&
   (\mathfrak{F}_{2,1}+\mathfrak{F}_{2,4}+\mathfrak{F}_{2,6})
   +\mathfrak{F}_{2,2}+\mathfrak{F}_{2,3}+\mathfrak{F}_{2,5}
   \\[2mm]
   &=&{\sigma\rho{\mu_2} (\gamma+\theta-1)\over \gamma+\theta}+\mathfrak{F}_{2,2}
   +\mathfrak{F}_{2,3}+\mathfrak{F}_{2,5}
   ={\sigma\rho{\mu_2} (1-\gamma-\theta)\over \gamma+\theta}=a_5+\kappa,
   \\[2mm]
   \mathfrak{F}_3&=&[\,\mathfrak{F}_{3,1}+(\mathfrak{F}_{3,3}
   +\mathfrak{F}_{3,4})\,]+\mathfrak{F}_{3,2}
   ={-(1-\gamma){\mu_2}^2\over 2(\gamma+\theta)}
   +{(1-\gamma){\mu_2}^2\over \gamma+\theta}
   ={(1-\gamma){\mu_2}^2\over 2(\gamma+\theta)}=a_6,
\end{eqnarray*}
with $a_4,a_5$ and $a_6$ defined in Theorem \ref{result2 of example 2}.

Noting $w(x,p,T)=x^{1-\gamma}/(1-\gamma)$, we know that $g_2(t)$ and $g_3(t)$ satisfy the following ODEs, respectively,
\begin{eqnarray}\nonumber
   &&{g_2^\prime(t)\over g_2(t)}
   +{\mu_0}(1-\gamma)+\kappa\overline{P}g_3(t)=0,\quad g_2(T)=1,
   \\[2mm]\label{ode1}
   &&g_3^\prime(t)+a_4 g_3^2(t)+a_5 g_3(t)+a_6=0,\quad g_3(T)=0.
\end{eqnarray}
It is clear that $g_2(t)$ takes the form of~\eqref{timesolution1}.

Next, we solve ODE \eqref{ode1} according to the following four cases:

{\bf Case 1. } If $\gamma>1$, then $a_4>0,a_6<0$ and $a_3<0<a_2$. So, $g_3(t)$ takes the form of \eqref{timesolution2} and $a_3<g_3\leq0$,
\begin{eqnarray}\label{eq27}
  0&\leq&g_2(t)\leq\exp\left\{[|\mu_0(1-\gamma)|+\kappa\overline{P}(|a_3|+|a_6|/\kappa)]T\right\},
  \\[2mm]\label{eq28}
  {|\phi^*|\over \sqrt{p}}&=&|\widehat{\phi}^*|\leq |\widehat{\phi}^*_1|+|\widehat{\phi}^*_2|
  \leq{\sigma\theta (|a_3|+|a_6|/\kappa)(1+\gamma+\theta)+\theta|\mu_2(1-\gamma)|\over |1-\gamma|(\gamma+\theta)},
  \\[2mm]\label{eq29}
  |\pi^*|&\leq&{\sigma|1-\gamma-\theta|(|a_3|+|a_6|/\kappa)+|\mu_2(1-\gamma)|\over|1-\gamma|(\gamma+\theta)}.
\end{eqnarray}

 {\bf Case 2. } If $\gamma<1$ and $\gamma+\theta>1$, then $a_4<0,a_6>0$ and $a_2<0<a_3$. Hence, $g_3(t)$ still takes the form of \eqref{timesolution2} and $0\leq g_3<a_3$. Moreover, the inequalities \eqref{eq27}-\eqref{eq29} still  hold.

{\bf Case 3. } If $\gamma+\theta=1$, then $a_4=0,a_5=-\kappa<0$ and $a_6>0$. Hence, $g_3(t)$ still takes the form of \eqref{timesolution2} and $0\leq g_3<-a_6/\kappa$. Moreover, the inequalities \eqref{eq27}-\eqref{eq29} still  hold.

 {\bf Case 4. } If $\gamma+\theta<1$, then $a_4>0$ and $a_6>0$. Since $\kappa\geq 2\sigma|\mu_2|/(\gamma+\theta)$, we have
$$
   a_5< -{\sigma|\mu_2| \over \gamma+\theta},\quad
   a_5^2> {\sigma^2{\mu_2}^2\over(\gamma+\theta)^2}
   \geq {\sigma^2{\mu_2}^2\over (\gamma+\theta)^2}(1-\gamma-\theta)
   [\,\rho^2+(1-\rho^2)(\gamma+\theta)\,]=4a_4a_6,
$$
and $a_2>a_3>0$. Hence, $g_3(t)$ still takes the form of \eqref{timesolution2} and $0\leq g_3<a_3$. Moreover, the inequalities \eqref{eq27}-\eqref{eq29} still  hold.

Until now, we have verified all assumptions in Theorem 5. Then, the conclusion in Theorem 6 can be obtained easily due to Theorem 5.
\hfill$\Box$\medskip

\bibliographystyle{bmc-mathphys}

\end{document}